\documentclass{llncs}

\usepackage[lncs,amsthm,autoqed]{pamas}
\usepackage{arydshln}

\usepackage{paralist}
\usepackage{txfonts}
\usepackage{array,xspace,multirow,epic}
\usepackage{natbib}
\usepackage{color}
\usepackage[boxed]{algorithm}
\usepackage{algorithmic}
\usepackage{booktabs}  
\usepackage{verbatim,ifthen}
\usepackage{enumitem}
\usepackage{pifont}
\usepackage{ifthen}
\usepackage{calrsfs,mathrsfs}
\usepackage{enumitem}
\usepackage{subfigure}

\usetikzlibrary{arrows}
\usetikzlibrary{decorations.pathreplacing}
\usetikzlibrary{patterns}

\newcommand\eat[1]{}

\setenumerate[1]{label=\rm(\it{\roman{*}}\rm),ref=({\it\roman{*}}),leftmargin=*}

\newlength{\wordlength}
\newcommand{\wordbox}[3][c]{\settowidth{\wordlength}{#3}\makebox[\wordlength][#1]{#2}}

\newcommand{\set}[1]{\{#1\}}
\newcommand{\midd}{\mathbin{:}}

\newcommand{\eqclass}[2][]{\ifthenelse{\equal{#1}{}}{[#2]}{[#2]_{\sim_{#1}}}}

\newcommand{\pref}{\succsim\xspace}
\newcommand{\Pref}[1][]{
	\ifthenelse{\equal{#1}{}}{\mathrel \succsim}{\mathop{\succsim_{#1}}}
}                                          
\newcommand{\sPref}[1][]{                  
	\ifthenelse{\equal{#1}{}}{\mathrel P}{\mathop{P_{#1}}}
}                                          
\newcommand{\Indiff}[1][]{                 
	\ifthenelse{\equal{#1}{}}{\mathrel I}{\mathop{I_{#1}}}
}
\newcommand{\prefset}[1][]{\ifthenelse{\equal{#1}{}}{\mathcal{R}}{\mathcal{R}_{#1}}}

\title{Cake Cutting Algorithms \\for Piecewise Constant and Piecewise Uniform Valuations\\
}

\author{%
	Haris Aziz\inst{1} \and Chun Ye\inst{2}}

\institute{%
NICTA and UNSW, Sydney Australia NSW 2033\\
	\email{haris.aziz@nicta.com.au}\and 
	Columbia University,
	New York, NY 10027-6902, USA \\
		\email{cy2214@columbia.edu}
	}

\begin{document}

\maketitle

\begin{abstract}
	Cake cutting is one of the most fundamental settings in fair division and mechanism design without money.
In this paper, we consider different levels of three fundamental goals in cake cutting: fairness, Pareto optimality, and strategyproofness. In particular, we present robust versions of envy-freeness and proportionality that are not only stronger than their standard counter-parts but also have less information requirements.  
We then focus on cake cutting with piecewise constant valuations and present three desirable algorithms: 
\emph{CCEA (Controlled Cake Eating Algorithm)}, \emph{MEA (Market Equilibrium Algorithm)} and \emph{CSD (Constrained Serial Dictatorship)}.
CCEA is polynomial-time, robust envy-free, and non-wasteful. It relies on parametric network flows and recent generalizations of the probabilistic serial algorithm. For the subdomain of piecewise uniform valuations, we show that it is also group-strategyproof.  
Then, we show that there exists an algorithm \emph{(MEA)} that is polynomial-time, envy-free, proportional, and Pareto optimal. MEA is based on computing a market-based equilibrium via a convex program and relies on the results of \citet{RePo98a} and \citet{DPSV08a}. Moreover, we show that MEA and CCEA are equivalent to mechanism 1 of Chen et. al. \cite{{CLPP12a}} for piecewise uniform valuations. 
We then present an algorithm \emph{CSD} and a way to implement it via randomization that satisfies strategyproofness in expectation, robust proportionality, and unanimity for piecewise constant valuations. For the case of two agents, it is robust envy-free, robust proportional, strategyproof, and polynomial-time.
Many of our results extend to more general settings in cake cutting that allow for variable claims and initial endowments. We also show a few impossibility results to complement our algorithms. The impossibilities show that the properties satisfied by CCEA and MEA are maximal subsets of properties that can be satisfied by any algorithm for piecewise constant valuation profiles.
%
\end{abstract}


\section{Introduction}


	Cake cutting is one the most fundamental topics in fair division~\citep[see \eg][]{Proc12a,BrTa96a,RoWe98a}. It concerns the setting in which a cake is represented by an interval $[0,1]$ and each of the $n$ agents has a value function over the cake. The main aim is to divide the cake fairly. 
	The framework is general enough to encapsulate the important problem of allocating a heterogeneous divisible good among multiple agents with different preferences.
	The cake cutting problem applies to many settings including the division of rent among housemates, disputed land between land-owners, and work among co-workers. It is especially useful in scheduling the use of a valuable divisible resource such as server time.

	Within the cake cutting literature, the most important criteria of a fair allocation are \emph{envy-freeness} and \emph{proportionality}.  In an envy-free allocation, each agent considers his allocation at least as good as any other agent's allocation. An envy-free allocation is guaranteed to exist~\citep[see \eg][]{Stro80a, Su99a}. In a proportional allocation, each agent gets at least $1/n$ of the value he assigns to the cake. A desirable aspect of envy-freeness is that it implies proportionality.\footnote{This statement holds with the additional assumption of \emph{full allocation}: that every portion of the cake that is desired by at least one agent is allocated to some agent.  Otherwise, the empty allocation satisfies envy-freeness, but not proportionality.}

	Computation of a fair allocation of cake is one of the fundamental problems in algorithmic economics. \citet{BrTa95a}  designed an envy-free cake cutting algorithm for an arbitrary number of players. Although their algorithm is guaranteed to eventually terminate, its running time is unbounded. Moreover, the algorithm can divide the cake into infinitely small segments. 
	Since the result of \citet{BrTa95a}, researchers have examined restricted value density functions and proposed envy-free algorithms for them.  
In order to ascertain the running time of a cake cutting algorithm, it is important to know the computational model and input to the problem. In most of the literature~\citep[see \eg][]{RoWe98a}, it is assumed that the value an agent ascribes to any segment of the cake can be queried or evaluated via an oracle. While the classical literature uses this query model, computer scientists recently looked at the problem from the point of view of full report, as is common in mechanism design.
Throughout the paper we focus on \emph{piecewise constant value density functions} and \emph{piecewise uniform value density functions}. 
Piecewise constant value density functions are one of the most fundamental class of value functions. Piecewise uniform valuations are a restricted class of piecewise constant valuations. 
	\citet{CLPP10a, CLPP12a} presented a discrete, strategyproof, polynomial-time, envy-free and Pareto optimal algorithm for piecewise uniform valuations. They stated that generalizing their results for piecewise constant valuations is an open problem. They also presented an envy-free and proportional algorithm that satisfies strategyproofness by resorting to randomization.


In this paper, we consider three of the most enduring goals in mechanism design and fair division: fairness, Pareto optimality and strategyproofness. Since many fair division algorithms need to be deployed on a large scale, we will also aim for algorithms that are \emph{computationally} efficient. Our main research question in this paper is as follows: \emph{among the different levels of fairness, Pareto optimality, strategyproofness, and efficient computability, what are the maximal set of properties that can be satisfied simultaneously for piecewise constant and piecewise uniform valuations?} Our main contribution is a detailed study of this question including the formulation of a number of desirable cake cutting algorithms satisfying many of the properties.

In the case where Pareto optimality cannot be satisfied, we also consider a weaker notion of efficient called \emph{non-wastefulness}. Non-wastefulness dictates that every portion of the cake that is desired by at least one agent is allocated to some agent who desires it.

For fairness, we not only consider the standard notions envy-freeness and proportionality but we also propose the concept of \emph{robust fairness} --- in particular \emph{robust envy-freeness} and \emph{robust proportionality}. The main idea of an allocation being robust envy-free is that even if an agent readjusts his value density function, as long as the ordinal information of the function is unchanged, then the allocation remains envy-free. The main advantages of robust envy-freeness are less information requirements and envy-freeness under uncertainty. 
\footnote{Although full information is a standard assumption in cake cutting, it can be argued that it is unrealistic that agents have exact utility value for each segment of the cake. Even if they do report exact utility values, they may be uncertain about these reports. Robust fairness bypasses these issues.}

For strategic properties, we consider three notions of truthfulness namely strategyproofness, weak group-strategyproofness and group-strategyproofness. In most of the cake-cutting literature, an algorithm is considered `strategyproof' if truth-telling is a maximin strategy~\citep{Bram08a} and it need not be dominant strategy incentive compatible. When we refer to strategyproofness, we will mean the reporting the truthful valuations is a dominant strategy. This stronger notion of strategyproofness has largely been ignored in cake-cutting barring a few recent exceptions~\citep{CLPP10a, CLPP12a, MoTa10a,MaNi12a}. 	\footnote{\citet{Proc12a} writes that the ``design of strategyproof cake cutting algorithms is largely an open problem.''}

	We present \emph{CCEA (Controlled Cake Eating Algorithm)} for piecewise constant valuations and show that it is polynomial-time and robust envy-free and robust proportional. 
	CCEA depends on a reduction to the  generalizations~\citep{AtSe11a,KaSe06a} of the \emph{PS (probabilistic serial)} algorithm introduced by \citet{BoMo01a} in the context of random assignments.\footnote{The CC algorithm of \citet{AtSe11a} is a generalization of the EPS algorithm \citep{KaSe06a} which in turn is a generalization of PS algorithm of \citet{BoMo01a}.} The algorithm relies on solving the parametric network flows~\citep[see \eg][]{GGT89a}.
	We show that the algorithm can handle variable claims and private endowments for piecewise constant valuations and also satisfies group-strategyproofness under piecewise uniform valuations.

		If we insist on Pareto optimality, then we show that there exists an algorithm which we refer to as the \emph{MEA (Market Equilibrium Algorithm)}  that is discrete, polynomial-time Pareto optimal, envy-free, and proportional for piecewise constant valuations. The algorithm relies on the Walras equilibrium formulation of \citet{RePo98a} for finding an $\alpha$-envy-free for general cake cutting valuations and the result of \citet{DPSV08a} that market equilibrium for Fischer markets with linear utilities can be computed in polynomial time. Both CCEA and MEA not only coincide on piecewise uniform valuations but are also group-strategyproof.
	
	Although CCEA and MEA are desirable algorithms, they are not strategyproof for piecewise constant valuations. %
		           We present another algorithm called \textit{CSD (Constrained Serial Dictatorship)} which is strategyproof in expectation, robust proportional, and satisfies unanimity. For the important case of two agents\footnote{Many fair division problems involve disputes between two parties.}, it is polynomial-time,  and robust envy-free. To the best of our knowledge, it is the first cake cutting algorithm for piecewise constant valuations that satisfies strategyproofness, proportionality, and unanimity at the same time. CSD requires some randomization to achieve strategyproof in expectation. However, CSD is discrete in the sense that it gives the same utility guarantee (with respect to the reported valuation functions) over all realizations of the random allocation. Although CSD uses some essential ideas of the well-known \emph{serial dictatorship} rule for discrete allocation, it is significantly more involved.

	

	%
	%
	%
	Our main technical results are as follows.


	\begin{theorem}\label{th:maintheorem1}
		For cake cutting with piecewise constant valuations, 
		there exists an algorithm (CCEA) that is discrete, polynomial-time, robust envy-free, and non-wasteful. 
		%
		\end{theorem}

			\begin{theorem}\label{th:maintheorem3}
		For cake cutting with with piecewise constant valuations, there exists an algorithm (MEA)  that is discrete, polynomial-time, Pareto optimal, and envy-free.
\end{theorem}

	\begin{theorem}\label{th:maintheorem2}
		For cake cutting with piecewise uniform valuations,
		there exists algorithms (CCEA and MEA) that are discrete, 
	group strategyproof, polynomial-time, robust envy-free and Pareto optimal. 
		\end{theorem}
		
		\begin{theorem}\label{th:maintheorem1b}
									For cake cutting with piecewise constant valuations, there exists a randomized implemention of an algorithm (CSD) that is (ex post) robust proportional, (ex post) symmetric, and (ex post) unanimous and strategyproof in expectation. For two agents, it is polynomial-time, robust proportional and robust envy-free.

			%
			\end{theorem}

\noindent Our positive results are complemented by the following impossibility results. These impossibility results suggest the properties satisfied by CCEA and MEA are maximal subsets of properties that can be satisfied by any algorithm for piecewise constant valuation profiles.
 
\begin{theorem} \label{PE fair impossibility}
For piecewise constant valuation profiles with at least two agents, there exists no algorithm that is both Pareto optimal and robust proportional.
\end{theorem}

\begin{theorem} \label{PE SP Prop impossibility}
For piecewise constant valuation profiles with at least two agents, there exists no algorithm that is strategyproof, Pareto optimal, and proportional. 
\end{theorem}

\begin{theorem} \label{SP Rob Prop Non W impossibility}
For piecewise constant valuation profiles with at least two agents, there exists no algorithm that is strategyproof, robust proportional, and non-wasteful. 
\end{theorem}

As a consequence of CCEA and MEA, we generalize the positive results regarding piecewise uniform valuations in \citep{CLPP10a,CLPP12a} and piecewise constant valuations in \citep{CLPP11a} in a number of ways such as handling richer cake cutting settings, handling more general valuations functions, achieving a stronger envy-free concept, or a stronger strategyproofness notion.  
Moreover, we prove that CCEA and MEA --- two different algorithms --- are equivalent in the domain of piecewise constant valuations. Furthermore, we show that both CCEA and MEA are generalizations of the main mechanism in \citep{CLPP12a, CLPP10a}.
We also show which combinations of properties are impossible to satisfy simultaneously.
Some of our main results are summarized in Table~\ref{table:summary:cake}.

		\begin{table}[t!]
		\centering
			\scalebox{0.8}{
	\centering
	\begin{tabular}{l|lccccccccccc}
	\toprule
	&Restriction&DISC&R-EF&EF&R-PROP&PROP&GSP&W-GSP&SP&PO&UNAN&POLYT\\
	Algorithms&&&\\
	\midrule
	CCEA&-&+&+&+&+&+&-&-&-&-&+&+\\
	CCEA&pw uniform&+&+&+&+&+&+&+&+&+&+&+\\

		\midrule
	MEA&&+&-&+&-&+&-&-&-&+&+&+\\
	MEA&pw uniform&+&+&+&+&+&+&+&+&+&+&+\\
	
	\midrule

		CMSD&-&-&-&-&+&+&-&-&+&-&+&-\\
		CMSD&pw uniform&-&-&-&+&+&-&-&+&+&+&-\\

		CMSD&2 agents&-&+&+&+&+&-&-&+&-&+&+\\
	%
	%
	%
%
	%
	

	\bottomrule
	\end{tabular}
	}

		\caption{Properties satisfied by the cake cutting algorithms for pw (piecewise) constant valuations:  DISC (discrete), R-EF (robust envy-freeness), EF (envy-freeness), R-PROP (robust proportionality), PROP (proportionality), GSP (group strategyproof), W-GSP (weak group strategyproof), SP (strategyproof), UNAN (unanimity), PO (Pareto optimal) and POLYT (polynomial-time). 
		}
	\label{table:summary:cake}
	\end{table}

After presenting our main results, we enrich the cake cutting domain in two ways. We allow agents to have initial endowments of the cake. Moreover we consider the more general setting in which agents may have different claims or entitlements to the cake. We show that many of our results carry over to these more general settings.

	\section{Preliminaries}

	\subsection{Cake cutting setting}
	We consider a cake which is represented by the interval $[0,1]$.
	A \emph{piece of cake} is a finite union of disjoint subintervals of $[0,1]$. The length of an interval $I=[x,y]$ is  $len(I)=y-x$.  As usual, the set of agents is $N=\{1,\ldots, n\}$. Each agent has a piecewise continuous \emph{value density function} $v_i:[0,1]\rightarrow [0,\infty]$. The value of a piece of cake $X$ to agent $i$ is $V_i(X)=\int_{X}v_i(x)dx=\sum_{I\in X}\int_Iv_i(x)dx$. As generally assumed, valuations are non-atomic ($V_i([x,x])=0$)
and additive: $V_i(X\cup Y)=V_i(X)+V_i(Y)$ where $X$ and $Y$ are disjoint. 
	The basic cake cutting setting can be represented by the set of agents and their valuations functions, which we will denote as \emph{a profile of valuations}. In this paper we will assume that each agent's valuation function is private information for the agent that is not known to the algorithm designer. Each agent reports his valuation function to the designer and the designer then decides how to make the allocations based on the reported valuations.

Later on we will also consider two important extensions of cake cutting: claims and private endowments.	We will assume that agents have the following \emph{claims} on the cake respectively: $c_1,\ldots, c_n$. In the original cake cutting problem agents have equal claims. 
	Each agent $i\in N$ has a \emph{private endowment} $\omega(i)$ which is a segment of the cake privately owned by $i$.  The cake is assembled by joining the pieces $\omega(1), \ldots, \omega(n)$.
	Therefore the cake cutting setting in its full generality can be represented as a quadruple $(N,v,\omega,c)$.

	An allocation is a partitioning of the cake into $n$ pieces of cake  $X_1,\ldots, X_n$ such that the pieces are disjoint (aside from the interval boundaries) and $X_i$ is allocated to agent $i\in N$. A cake cutting algorithm takes as input $(N,v,\omega,c)$ and returns an allocation. 

	\subsection{Preference functions}

In this paper we will only consider \emph{piecewise uniform} and \emph{piecewise constant} valuations functions. A function is \emph{piecewise uniform} if the cake can be partitioned into a finite number of intervals such that for some constant $k_i$, either $v_i(x)=k_i$ or $v_i(x)=0$ over each interval. A function is \emph{piecewise constant} if the cake can be partitioned into a finite number of intervals such that $v_i$ is constant over each interval. In order to report his valuation function to the algorithm designer, each agent will specify a set of points $\{d_1, ... , d_m\}$ that represents the consecutive points of discontinuity of the agent's valuation function as well as the constant value of the valuation function between every pair of consecutive $d_j$'s.


	For a function $v_i$, we will refer by $\hat{V_i}=\{v_i'\midd v_i(x)\geq v_i(y)>0 \iff  v_i'(x)\geq v_i'(y)>0 ~~\forall x,y\in [0,1]\}$ as the set of density functions \emph{ordinally equivalent} to $v_i$. 
	Note that if an algorithm takes as input $v_i$ and returns the same output for any $v_i'\in \hat{V_i}$, then it is oblivious to the exact cardinal information of $v_i$.

	\subsection{Properties of allocations}
	%
           An \emph{allocation} is a partition of the interval $[0,1]$ into a set $\{X_1,\ldots, X_n, W\}$, where $n$ is the number of agents and $X_i$ is a piece of cake that is allocated to agent $i$. And $W$ is the piece of the cake that is not allocated. All of the fairness and efficiency notations that we will discuss next are with respect to the reported valuation functions.
	In an \emph{envy-free allocation}, we have $ V_i(X_i)\geq V_i(X_j)$ for each pair of agent $i,j\in N$,. An allocation is \emph{individually rational} if $V_i(X_i)\geq V_i(\omega(i))$. In a \emph{proportional} allocation, each agent gets at least $1/n$ of the value he has for the entire cake. An allocation satisfies \emph{symmetry or equal treatment of equals} if any two agents with identical valuations get same utility. Clearly, envy-free implies proportionality and also symmetry. An allocation $X$ is \emph{Pareto optimal} if no agent can get a higher value without some other agent getting less value. Formally, $X$ is Pareto optimal if there does not exists another allocation $Y$ such that $V_i(Y_i)\geq V_i(X_i)$ for all $i\in N$ and $V_i(Y_i)> V_i(X_i)$ for some $i\in N$. For any $S \subseteq [0,1]$, define $D(S) = \{i \in N | V_i(S) > 0\}$. An allocation $X$ is \emph{non-wasteful} if for all $S \subseteq [0,1]$, $S \subseteq \cup_{i \in D(S)} X_i$. In other words, an allocation is non-wasteful if every portion of the cake desired by at least one agent is allocated to some agent who desires it.



	We now define robust analogues of the fairness concepts defined above. 
		An allocation satisfies \emph{robust proportionality} if for all $i,j\in N$ and for all $v_i'\in \hat{V_i}$, $\int_{X_i}v_i'(x)dx\geq 1/n\int_{0}^{1}v_i'(x)dx$.
		An allocation satisfies \emph{robust envy-freeness} if for all $i,j\in N$ and for all $v_i'\in \hat{V_i}$, $\int_{X_i}v_i'(x)dx\geq \int_{X_j}v_i'(x)dx$.
			Notice that both robust envy-freeness and robust proportionality would require each agents to get a piece of cake of the same length if every agent desires the entire cake.
		



	We give an example of piecewise constant value density function and demonstrate how the standard concept of envy-freeness is not robust under uncertainty. 

	\begin{example}[A cake-cutting problem with piecewise constant valuations]\label{example:piecewise}
	Consider the cake cutting problem in Figure~\ref{figure:piecewise}.
	An allocation in which both agents get regions in which their value density function is the highest is envy-free. Agent 1 gets utility one for his allocation and has the same utility for the allocation of agent 2.  However if its probability density function is slightly lower in region $[0,0.1]$, then agent 1 will be envious of agent 2.
	\end{example}
\vspace{-3em}
	 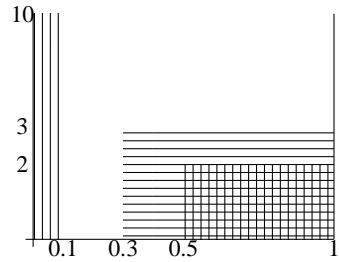
\begin{figure}
	 \centering
	 	\[	
	 			\begin{tikzpicture}[scale=0.2]
	 				\centering
	 				\draw[-] (0,-.5) -- (0,15);     
	 				\draw[-] (-.5,0) -- (20,0);

	 				\draw[-] (20,15) -- (20,0);


	 							\fill[pattern=vertical lines,] (0,0) rectangle (2,15);
	 							\fill[pattern=vertical lines] (10,0) rectangle (20,5);
	 							\fill[pattern=horizontal lines] (6,0) rectangle (20,7.5);

	 							\draw (2,-.6) node(c){\small $0.1$};
	 							\draw (6,-.6) node(c){\small $0.3$};
	 							\draw (20/2,-.6) node(c){\small $0.5$};
	 						%
	 							\draw (20,-.6) node(c){\small$1$};

	 				\draw(-.7,15) node(z){\small $10$};
	 				\draw(-.7,7.5) node(z){\small $3$};
	 				\draw(-.7,5) node(z){\small $2$};

	 				\end{tikzpicture}
	 		\]
	 \caption{Example of a cake cutting problem with piecewise constant value density functions. The area with vertical lines is under the value density function of agent $1$ and the area with horizontal lines is under the value density function of agent $2$.} 
	\label{figure:piecewise}
	 \end{figure}

	 \begin{figure}
	 \centering
	 	\[	
	 			\begin{tikzpicture}[scale=0.20]
	 				\centering
	 				\draw[-] (0,-.5) -- (0,15);     
	 				\draw[-] (-.5,0) -- (20,0);

	 				\draw[-] (20,15) -- (20,0);


	 							\fill[pattern=vertical lines,] (0,0) rectangle (2,10);
	 							\fill[pattern=vertical lines] (10,0) rectangle (20,5);
	 							\fill[pattern=horizontal lines] (6,0) rectangle (20,10);

	 							\draw (2,-.6) node(c){\small $0.1$};
	 							\draw (6,-.6) node(c){\small $0.3$};
	 							\draw (20/2,-.6) node(c){\small $0.5$};
	 						%
	 							\draw (20,-.6) node(c){\small$1$};

	 				\draw(-.7,15) node(z){\small $10$};
	 				\draw(-.7,7.5) node(z){\small $3$};
	 				\draw(-.7,5) node(z){\small $2$};

	 				\end{tikzpicture}
	 		\]
	 \caption{Example of a cake cutting problem with piecewise constant value density functions. The area with vertical lines is under the value density function of agent $1$ and the area with horizontal lines is under the value density function of agent $2$. The valuation functions of both agents are ordinally equivalent to the ones in Figure~\ref{figure:piecewise}.}
	\label{figure:piecewise-b}
	 \end{figure}
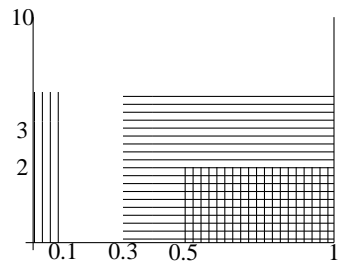

	\newpage	
	\subsection{Properties of cake cutting algorithms}
	A \emph{deterministic cake cutting algorithm} is a mapping from the set of valuation profiles to the set of allocations. A \emph{randomized cake cutting algorithm} is a mapping from the set of valuation profiles to a space of distributions over the set of allocations. The output of the algorithm in this case for a specific valuation profile is a random sample of the distributional function over the set of allocation for that profile. 

	An algorithm (either deterministic or randomized) satisfies property $P$ if it always returns (with probability 1) an allocation that satisfy property $P$. A deterministic algorithm is \emph{strategyproof} if no agent ever has an incentive to misreport in order to get a better allocation. The notion of strategyproofness is the one well-established in social choice~\citep[see \eg][]{Gibb73a}  and much stronger than the notion of `strategyproofness' used in some of the cake-cutting literature~\citep[see \eg][]{Bram08a}.
	By strategyproof, we mean truthful as has been used in \citep{CLPP10a}. 
	Similarly, a deterministic algorithm is \emph{group-strategyproof} if there exists no coalition $S\subseteq N$ such that members of $S$ can misreport their preferences so that each agent in $S$ gets as preferred an allocation and at least one agent gets a strictly better payoff. A deterministic algorithm is \emph{weak group-strategyproof} if there exists no coalition $S\subseteq N$ such that members of $S$ can misreport their preferences so that each agent in $S$ gets a strictly more preferred allocation. 
	A randomized algorithm is \emph{strategyproof in expectation} if the expected utility from the random allocation that every agent receives in expectation under a profile where he reported truthfully is at least as large as the expected that he receives under a profile where he misreports while fixing the other agents' reports.

	We say that a cake cutting algorithm satisfies \emph{unanimity}, if when each agent's most preferred $1/n$ length of the cake is disjoint from another agent, $1/n$ length of the cake, then each agent gets their most preferred piece of cake of length $1/n$.
	


	\subsection{Relation between the properties of cake cutting algorithms}

	In this subsection, we recap the main properties of cake cutting algorithms: \begin{inparaenum}[i)]
\item proportionality,
\item robust proportionality,
\item envy-freeness,
\item robust envy-freeness,
\item symmetry,
\item non-wastefulness,
\item Pareto optimality,
\item unanimity
\item strategyproofness,
\item weak group strategyproofness,
\item group strategyproofness,
\item polynomial-time.
	\end{inparaenum}

\begin{remark}
For cake cutting, 
\begin{inparaenum}[\itshape a\upshape)]	
\item Envy-freeness and non-wastefulness $\implies$ proportionality;
\item Robust proportionality $\implies$ proportionality;
\item Robust envy-freeness $\implies$ envy-freeness;
\item Robust envy-freeness and non-wastefulness $\implies$ robust proportionality; 
\item Group strategyproofness $\implies$ weak group strategyproofness $\implies$ strategyproofness;
\item Pareto optimality $\implies$ non-wastefulness;
\item Pareto optimality $\implies$ unanimity;
\item two agents, proportionality $\implies$ envy-freeness;
\item two agents, robust proportionality $\implies$ robust envy-freeness.
	\end{inparaenum}

\end{remark}


\subsection{The free disposal assumption}
We may assume without lost of generality that every part of the cake is desired by at least one agent.  If that is not the case, then we can discard the parts that are desired by no one and rescale what is left so that we get a $[0,1]$ interval representation of the cake. Notice that this procedure preserves the aforementioned properties of fairness, efficiency and truthfulness. The free disposal assumption that we are making is necessary to ensure strategyproofness for piecewise uniform valuation functions. See \citep{CLPP12a} for a discussion on the necessity of this assumption. \\


\noindent Before we present our algorithms, we will first take a detour to the literature on random assignments, as some of the algorithms in the random assignment literature are closely related to our algorithms.

	\section{A detour to random assignments}
	An assignment problem is a triple $(N, H, \pref)$ such that $N=\{1,\ldots, n\}$ is a set of agents, $H=\{h_1, \ldots, h_n\}$ is a set of houses and $\pref=(\pref_1,\ldots, \pref_n)$ is the preference profile in which $\pref_i$ denotes the preferences of agent $i$ over houses $H$. A \emph{deterministic assignment} is a one-to-one mapping from $N$ to $H$. 
	A \emph{random allocation} is a probability distribution over $H$. 
	A random assignment $p$ gives a random allocation to each agent. It can be represented by a bistochastic matrix $p$ in which the $i$th row is denoted by $p_i$ and all $i\in A$, and $h\in H$, $p_{ih}\geq 0$, $\sum_{j\in A}p_{jh}=\sum_{h'\in H}p_{ih'}=1$.\footnote{The bistochasticity of the matrix $p$ holds when there are the same number of agents as there are objects, which can be assumed without lost of generality by adding dummy agents or objects} The term $p_{ih}$ denotes the probability with which agent $i$ gets house $h$.
An assignment problem has commonalities with cake cutting with piecewise constant valuations. They also have some fundamental differences. For example, in cake cutting, the agents do not have continuous constant valuations over pre-determined segments of the cake.

	Given two random assignments $p$ and $q$, $p_i \succsim_i^{SD} q_i$ i.e.,  a player $i$ weakly SD prefers $p_i$ to $q_i$ if for all $h$,
	$	\sum_{h_j\in \set{h_k\midd h_k\succsim_i h}}p_{ih_j} \ge \sum_{h_j\in \set{h_k\midd h_k\succsim_i h}}q_{ih_j}.$ 
	Another way to see the SD relation is as follows. A player $i$ weakly SD prefers allocation $p_i$ to $q_i$ if for all vNM utilities consistent with his ordinal preferences, $i$ gets at least as much expected utility under $p_i$ as he does under $q_i$.
	Furthermore $p \succsim^{SD} q$, i.e., $p$ \emph{stochastically dominates} $q$ if $p_i \succsim_i^{SD} q_i$ for all $i\in N$ and $q\neq p$. 

	An algorithm satisfies \emph{SD-efficiency} if each returned assignment is Pareto optimal with respect to the SD-relation~\citep[see \eg][]{Yilm09a}. 
	An algorithm satisfies \emph{SD envy-freeness} if each agent (weakly) SD prefers his allocation to that of any other agent. SD envy-freeness is a very demanding notion of fairness. The reader may be able to notice that our notion of robust envy-freeness in cake cutting is based on a similar idea as SD envy-freeness. 
	We will consider random allocations as fractional allocations and random assignments as fractional assignments. Viewing the probability of getting a house simply as getting a fraction of the house is especially useful when some houses are not complete but only partial. In this vein, the definition of SD dominance should also be considered from the perspective of fractional allocations rather than probability distributions.

	%


	The most basic assignment problem concerns $n$ agents having strict preferences over $n$ objects. 
	For this basic setting a simple yet ingenious \emph{PS (probabilistic serial)} algorithm introduced by \citet{BoMo01a} and ~\citet{CrMo01a} which uses the \emph{simultaneous eating algorithm (SEA)}. Each house is considered to have a divisible probability weight of one, and agents simultaneously and with the same eating rate consume the probability weight of their most preferred house until no house is left. 
	The random allocation allocated to an agent by PS is then given by the amount of each object he has eaten until the algorithm terminates. The main result of \citet{BoMo01a} was that the fractional assignment returned by the PS algorithm is SD envy-free and SD-efficient.

	The PS algorithm has been extended in various ways. The EPS (extended PS algorithm) of \citet{KaSe06a} generalized PS to the case for indifferences using parametric network flows. EPS also generalized the \emph{egalitarian rule} of \citet{BoMo04a} for dichotomous preferences. \citet{Yilm09a} and \citet{AtSe11a} extended the work of \cite{KaSe06a} to propose PS generalization which also takes care of private endowments where $e_i$ indicates the endowment of agent $i$. For the case of endowments, \citet{Yilm09a} introduced the idea of justified envy-freeness. An assignment $p$ satisfies justified envy-freeness if for all $i,j\in N$,  $p_i \mathrel{\succsim_i^{SD}} p_{j}$ or $\neg(p_i \mathrel{\succsim_{j}^{SD}} e_{j}.)$
	The algorithms in \citep{Yilm09a,AtSe11a} satisfy justified envy-freeness in the presence of private endowments. 
	In our algorithm CCEA, we rely on the full power of the controlled-consuming (CC) algorithm of \citet{AtSe11a}  which combines almost all the desirable features of other extensions of PS. In particular, we use the following fact. \emph{There exists an extension of the PS algorithm which can simultaneously handle indifferences in preferences, unacceptable objects in preferences, allocation of multiple objects to agents, agent owning fractions of houses, partial houses being available, and still returns an assignments which satisfies SD justified envy-freeness and SD-efficiency.} In addition, if there are no private endowments, then the extension can also handle variable eating rates.
	The Controlled-Consuming (CC) algorithm of \citet{AtSe11a}  can handle the case where each agent owns fractions of the complete houses. 
	We also require that for some houses, only an arbitrary fraction of the house is available to the market. This can be handled by a modification to CC~\citep[page 30, ][]{AtSe11a}. Finally we require the agents to want to be allocated as many houses as possible. This does not require any modification to CC.
	In the absence of endowments but presence of variable eating rates, CC is equivalent to the EPS algorithm that can also cater for variable eating rates~\citep[Section 6.4, ][]{KaSe06a}. 

	\section{CCEA --- Controlled Cake Eating Algorithm}

	CCEA is based on CC (Controlled Consuming) algorithm of \citet{AtSe11a}. Since the original PS algorithm utilized the simultaneous eating algorithm, hence the name  \emph{Controlled Cake Eating Algorithm.}
	CCEA first divides the cake up into disjoint intervals each of whose endpoints are consecutive points of discontinuity of the agents' valuation functions. We will refer to these intervals as \emph{intervals induced by the discontinuity points}. The idea is to form a one-to-one correspondence of the set of cake intervals with a set of houses of an assignment problem. Since intervals may have different lengths, we consider the house corresponding to the interval with the maximum length as a complete house where as intervals corresponding to other houses are partial houses. The preferences of agents over the houses are naturally induced by the relative height of the piecewise constant function lines in the respective intervals. If an agent $i$ owns a sub-interval $J_j$, then in the housing setting, $e_{ih_j}$ is set to $size(h_j)$ and not to one. The reason is that an agent can only own as much of the house as exists. The technical heart of the algorithm is in CC (Controlled Consuming) algorithm of \citet{AtSe11a}. We recommend the reader to Section 3.2 of \citep{AtSe11a} in which an illustrative example on CC is presented. Once CC has been used to compute a fractional assignment $p$, it is straightforward to compute a corresponding cake allocation. If an agent $i$ gets a fraction of house $h_j$, then in the cake allocation agent $i$ gets the same fraction of subinterval $J_j$.

	\begin{algorithm}[h]
	  \caption{ Controlled Cake Eating Algorithm (CCEA) to compute a robust envy-free allocation for piecewise constant value functions. 
	}
	  \label{RAR-algo}
	\renewcommand{\algorithmicrequire}{\wordbox[l]{\textbf{Input}:}{\textbf{Output}:}} 
	 \renewcommand{\algorithmicensure}{\wordbox[l]{\textbf{Output}:}{\textbf{Output}:}}
	\begin{algorithmic}
		\REQUIRE  Piecewise constant value functions.
		\ENSURE  A robust envy-free allocation.
	\end{algorithmic}
	\algsetup{linenodelimiter=\,}
	  \begin{algorithmic}[1] 
			\normalsize
	\STATE Divide the regions according to agent value functions. Let $\mathcal{J}=\{J_1, \ldots, J_m\}$ be the set of subintervals of $[0,1]$ formed by consecutive marks.
	\STATE Consider $(N',H,\succsim , size(\cdot))$ where

	\begin{itemize}	
	\item $H=\{h_1,\ldots, h_m\}$ where $h_i=J_i$ for all $i\in \{1,\ldots, m\}$ 
	\item  $\pref$ is defined as follows: $h\succsim_i h'$ if and only if $v_i(x)\geq v_i(y)$ for $x\in h$ and $y\in h'$; \\
	$\pref_{n+1}$ is defined in the way that each house in the market is unacceptable to agent $n+1$.
	\item $size(h_j)=1$ for $h_j\in \arg \max_{j\in \{1,\ldots, m\}}(len(J_j))$;\\

	 $size(h_j)=\frac{len(J_j)}{(\max_{j\in \{1,\ldots, m\}}(len(J_j)))}$ \\for all $h_j\notin \arg \max_{j\in \{1,\ldots, m\}}(len(J_j))$;
	\end{itemize}

	\STATE $p\longleftarrow CC(N',H,\succsim, size(\cdot))$
	\STATE For interval $J_j$, agent $i$ is an allocated subinterval of $J_j$, denoted by $J^{i}_j$, which is of length $p_{ih_j}/size(h_j)\times len(J_j)$. For example, if $J_j = [a_i, b_i]$, then one possibility of $J^{i}_{j}$ can be $[a_i + \sum_{n=1}^{i-1} p_{ih_j}/size(h_j)\times len(J_j), \ a_i + \sum_{n=1}^{i} p_{ih_j}/size(h_j)\times len(J_j) ]$.  \\
	$X_i\longleftarrow \bigcup_{j=1}^m J^{i}_j$ for all $i\in N$
	  \RETURN $X=(X_1,\ldots, X_n)$ 
	 \end{algorithmic}
	\label{algo:f}
	\end{algorithm}

	In Example~\ref{example:transform}, we show how CCEA transforms a cake cutting problem with piecewise constant valuations into a random assignment problem.

	\begin{example}[Illustration of CCEA]\label{example:transform}
	We examine how CCEA runs on the cake cutting problem in Figure~\ref{figure:piecewise}. 
	Firstly,  let $\mathcal{J}=\{J_1, \ldots, J_4\}$ be the set of subintervals of $[0,1]$ formed by consecutive points of discontinuity are identified: $J_1=[0,0.1], J_2=[0.1,0.3], J_3=[0.3,0.5]$, and $J_4=[0.5,1]$. $J_2$ is discarded because it is desired by no agent. In set $\{h_1, h_3,h_4\}$, each house $h_j$ corresponds to subinterval $J_j$. The preferences of the agents over $H$ are inferred from their valuation function height in the subintervals so that $h_1 \succ_1 h_4 \succ_1 h_3$ and $h_3 \sim_2 h_4 \succ_2 h_1.$ We also set the number of units of each house that is available. Since $J_4$ is the biggest interval, we consider $h_4$ as complete house. So, $size(h_4)=1$, $size(h_1)=0.2$, and $size(h_3)=0.4$. 
	If we run CC over the housing market instance with the specified set of agents, houses, fraction of houses available to the market, and agent preferences, then the assignment returned by CC is as follows:
	$p_{1h_1}=0.2$, $p_{1h_3}=0$, $p_{1h_4}=0.6$, $p_{2h_1}=0$, $p_{2h_3}=0.4$, and $p_{2h_4}=0.4$. The house assignment $p$ can be used to divide the subintervals among the agents: $X_1=\{[0,0.1], [0.7,1]\}$ and $X_2=\{[0.3,0.5], [0.5,0.7]\}$.
	\end{example}

CCEA satisfies the strong fairness property of robust envy-freeness.

	\begin{proposition}\label{prop:CCEA is robust envy-free}
	For piecewise constant valuations, CCEA satisfies robust envy-freeness and non-wastefulness.
	\end{proposition}

	Let $m$ be the number of relevant subintervals in a cake cutting problem with piecewise constant valuations.

	\begin{proposition}\label{prop:CCEA time}
	CCEA runs in time $O(n^5m^2\log(n^2/m))$, where $n$ is the number of agents and $m$ is the number of subintervals defined by the union of discontinuity points of the agents' valuation functions.
	\end{proposition}

	Although CCEA satisfies the demanding property of robust envy-freeness, it is not immune to manipulation. We show that CCEA is not strategyproof even for two agents. In the next section, we will present a different algorithm that is both robust envy-free and strategyproof for two agents. 

	\begin{proposition}\label{prop:CCEA is not SP}
	For piecewise constant valuations, CCEA is not strategyproof even for two agents.
	\end{proposition}

If we restricted the preferences to piecewise uniform with no private endowment or variable claims, then CCEA is not only strategyproof but group-strategyproof. We first show that in this restricted setting, CCEA is in fact equivalent to the algorithm of \citep{CLPP10a}.

	\begin{lemma} \label{prop: equivalence}
	 For piecewise uniform value functions with no private endowments and variable claims, CCEA is equivalent to Mechanism 1 of \citet{CLPP10a}. 
	\end{lemma}

Since the set of valuations that can be reported is bigger in cake cutting than in the assignment problem, establishing group strategyproofness does not follow automatically from group-strategyproofness of CC for dichotomous preferences~\citep[Theorem 2, ][]{BoMo04a}. Using similar arguments, we give a detailed proof that CCEA and hence Mechanism 1 of \citet{CLPP10a} is group strategyproof for piecewise uniform valuations.\footnote{\citet{CLPP10a} had shown that their mechanism for piecewise uniform valuations is strategyproof.} In Section~\ref{sec:extensions}, we extend the result to the case where agents may have variable claims. 

\begin{proposition}\label{prop:gsp1}
For cake cutting with piecewise uniform value functions, CCEA is group strategyproof.
	\end{proposition}


%
%
	
	
	For piecewise uniform valuations, CCEA is also Pareto optimal. The result follows directly from lemma \ref{prop: equivalence} along with the fact that Mechanism 1 of \citet{CLPP10a} is Pareto optimal.
	
%
 
	\begin{proposition}\label{prop: pareto optimality}
	For cake cutting with piecewise uniform value functions, CCEA is Pareto optimal. 
	\end{proposition}




\section{MEA --- Market Equilibrium Algorithm}

In the previous section we presented CCEA which is not Pareto optimal for piecewise constant valuations. 
It turns out that if we relax the robust notion of fairness to envy-freeness, then we can use fundamental results in general equilibrium theory and recent algorithmic advances ~\citep{DPSV08a} to formulate a convex program that always returns an envy-free and Pareto optimal allocation as its optimal solution.  For each valuation profile, let $J = \{J_1, \ldots, J_k\}$ be the intervals whose endpoints are consecutive points in the union of break points of the agents' valuation functions. Let $x_{ij}$ be the length of any subinterval of $J_i$ that is allocated to agent $j$. Then we run a convex program to compute a Pareto optimal and envy-free allocation.
We will call the convex program outlined in Algorithm~\ref{algo:MARKET} as the \emph{Market Equilibrium Algorithm (MEA)}. MEA is based on computing the market equilibrium via a primal-dual algorithm for a convex program that was shown to be polynomial-time solvable by \citet{DPSV08a}. Notice that if we ignore strategyproofness, or in other words, if we assume that agents report truthfully, then agents are truly indifferent between which subinterval they receive since their valuation function is a constant over any $J_i$. Hence, one we determine the length of $J_j$ to be allocated to an agent, we can allocate any subinterval of that length to the agent.

	\begin{algorithm}[h!]
	  \caption{ The Market 	Equilibrium Algorithm to compute a Pareto optimal, envy-free, and proportional allocation.}
	  \label{RAR-algo}
	\renewcommand{\algorithmicrequire}{\wordbox[l]{\textbf{Input}:}{\textbf{Output}:}} 
	 \renewcommand{\algorithmicensure}{\wordbox[l]{\textbf{Output}:}{\textbf{Output}:}}
\begin{algorithmic}

		\REQUIRE  Cake-cutting problem with piecewise constant valuations.
		\ENSURE  A proportional, envy-free, and Pareto optimal allocation.
	\end{algorithmic}
	\algsetup{linenodelimiter=\,}
	  \begin{algorithmic}[1] 

\STATE Let $J = \{J_1, \ldots, J_k\}$ be the intervals whose endpoints are consecutive points in the union of break points of the agents' valuation functions. 
\STATE Let $x_{ij}$ be the length of any subinterval of $J_i$ that is allocated to agent $j$.
\STATE $l_i\longleftarrow len(J_i)$
\STATE Solve the following convex program. 
\begin{align}
\text{min} \quad
-\sum_{j =1}^{n} log(u_j) \notag \\
\text{s.t.} \quad
u_j = \sum_{i=1}^{k} v_{ij}x_{ij} &\quad \forall j = 1, \ldots, n \notag \\
\sum_{j=1}^{n} x_{ij} \leq l_i &\quad \forall i = 1, \ldots, k \notag \\
x_{ij} \geq 0 &\quad \forall i, j. \notag \notag
\end{align}
\STATE Let $u^{\star}_j$, $x^{\star}_{ij}$ be an optimal solution to the convex program. Partition every interval $J_i$ into $n$ subintervals where the $j$-th subinterval $J^{j}_i$ has length $x^{\star}_{ij}$.
\STATE $Y_j \longleftarrow \cup_{i=1}^{k} J^{j}_i$ be the allocation of each $j = 1, \ldots, n$.
	  \RETURN $Y=(Y_1,\ldots, Y_n)$.
	 \end{algorithmic}
	\label{algo:MARKET}
	\end{algorithm}

\begin{proposition} \label{PE, EF, Prop}
MEA is polynomial-time, Pareto efficient and envy free.
\end{proposition}

We mention here that the connection between a fair and efficient algorithm for cake cutting and computing market equilibria was already pointed by \citet{RePo98a}. \citet{RePo98a} presented an algorithm to compute an approximately envy-free and Pareto optimal allocation for cake cutting with general valuations. However their algorithm is not polynomial-time even for piecewise constant valuations~\citep{ZDO+10a}. 

MEA requires the machinery of convex programming. It remains open whether MEA can be implemented via linear programming.
\citet{CLPP11a} presented an algorithm that uses a linear program to compute an optimal envy-free allocation. The allocation is Pareto optimal among all envy-free allocations. However it need not be Pareto optimal in general. 


Although MEA is not robust envy-free like CCEA, it is Pareto optimal because it maximizes the Nash product.
What is interesting is that under uniform valuations, both MEA and CCEA are equivalent.
In the next result we demonstrate this equivalence~(Proposition~\ref{CCEA equivalence}). The proof requires a careful comparison of both CCEA and MEA under uniform valuations.

\begin{proposition} \label{CCEA equivalence}
For piecewise uniform valuations,  the allocation given by CCEA is identical to that given by MEA.
\end{proposition}
\begin{corollary}
For piecewise uniform valuations, MEA is group-strategyproof.
\end{corollary}



Thus if we want to generalize Mechanism 1 of \citet{CLPP10a} to piecewise constant valuations and maintain robust envy-freeness then we should opt for CCEA. On the other hand, if want to still achieve Pareto optimality, then MEA is the appropriate generalization. In both generalization, we lose strategyproofness.

%

%
%

\section{Impossibility Results}

Thus far, we presented two polynomial time algorithms, each of which satisfies a different set of properties. CCEA is robust envy-free and non-wasteful, whereas MEA is Pareto optimal and envy-free. This naturally leads to the following question: does there exist an algorithm satisfies all of the properties that CCEA and MEA satisfy? It turns out that the answer is no, as Theorem~\ref{PE fair impossibility} shows that there is no algorithm that is both Pareto efficient and robust proportional. 
Similarly, Theorem~\ref{PE SP Prop impossibility} argues that there is no algorithm that satisfies the properties CCEA satisfies along with strategyproofness. Lastly, Theorem~\ref{SP Rob Prop Non W impossibility} argues that there is no algorithm that satisfies the properties CCEA satisfies plus strategyproofness. The impossibility results are summarized in Table~\ref{table:summary:impossible}.

	\begin{table}[h!]
	\centering
		\scalebox{0.9}{
\centering
\begin{tabular}{ll}
\toprule
Impossibility& Reference\\
\midrule
Pareto efficient and robust proportional&Theorem.~\ref{PE fair impossibility}\\
Strategyproof, Pareto optimal, and proportional&Theorem~\ref{PE SP Prop impossibility}\\
Strategyproof, robust proportional, and non-wasteful 	& Theorem~\ref{SP Rob Prop Non W impossibility}\\
\bottomrule
\end{tabular}
}

	\caption{Impossibility results for cake-cutting algorithms for piecewise constant valuations.
	}
\label{table:summary:impossible}
\end{table}

Consequently, we may conclude that the properties satisfied by CCEA and MEA are respectively maximal subsets of properties that an algorithm can satisfy for piecewise constant valuations.

\section{CSD --- Constrained Serial Dictatorship Algorithm}

In the previous sections, we saw that CCEA and MEA are only strategyproof for piecewise uniform valuations. In light of the impossibility results established in the preivous section, it is reasonable to ask what other property along with strategyproofness can be satisfied by some algorithm. It follows from \citep[Theorem 3, ]{Schu97a} that the only type of strategyproof and Pareto optimal mechanisms are dictatorships. \citet{CLPP12a} raised the question whether there exists a strategyproof and proportional algorithm for piecewise constant valuations. The algorithm CSD answers this question partially.

Before diving into the CSD algorithm, it is worth noting that there is some fundamental difference between random assignment setting and the cake cutting setting. In the random assignment setting, the objects that we are allocating are well defined and known to the public. On the other hand, in the cake cutting setting, the discontinuity points of each agent's valuation function is private information for the agent. Hence, any algorithm that uses the reported discontinuity points to artificially create the objects runs into the risk of having the objects created by the algorithm be manipulated by the reports of the agents. In order to illustrate this difficulty, consider the uniform allocation rule.  The uniform allocation rule (that assigns $1/n$ of each house)~\citep{Cham04a} is both strategyproof and proportional in the random assignment setting. However it cannot be adapted for cake cutting with piecewise constant valuations since strategyproofness is no longer satisfied if allocating $1/n$ of each interval (induced by the agent valuations) is done deterministically.

\begin{proposition} \label{prop: uniform allocation is not SP}
The uniform allocation rule (done deterministically) is not strategyproof.
\end{proposition}

Now we are ready to present CSD. In order to motivate CSD, we will give a randomized algorithm that is strategyproof and robust proportional in expectation. The algorithm is a variant of random dictatorship: each agent has uniform probability of being chosen as a dictator. However, if the whole cake is acceptable to each agent, then each time a dictator is chosen, he will take the whole cake. This approach is not helpful since we return to square one of having to divide the whole cake. We add an additional requirement which is helpful. We require that each time a dictator is chosen, the piece he takes has to be of maximum value $1/n$ length of the total size of the cake. We will call this algorithm Constrained Random Serial Dictatorship (CRSD). Formally speaking, CRSD draws a random permutation of the agents. The algorithm then makes the allocation to agents in the order that the lottery is drawn. Everytime when it is agent $i$'s turn to receive his allocation, CRSD looks at the remaining portion of the cake and allocates a maximum value $1/n$ length piece of the cake to agent $i$ (break ties arbitrarily). Notice that CRSD is strategyproof,  as in every draw of lottery, it is in the best interest of the agents to report their valuation function truthfully in order to obtain a piece that maximizes his valuation function out of the remaining pieces of cake. Later on we will see, through the proof of Proposition~\ref{prop: CSD is robust prop}, that CRSD is robust proportional in expectation.  \\

CSD is an algorithm that derandomizes CRSD by looking at its allocation for all $n!$ different permutations and aggregate them in a suitable manner. The algorithm is formally presented as Algorithm~\ref{CSD-algo}.\\


\begin{algorithm}[t]
	  \caption{CSD (Constrained Serial Dictatorship)---  proportional and unanimous algorithm for piecewise constant valuations}
	  \label{CSD-algo}
	\renewcommand{\algorithmicrequire}{\wordbox[l]{\textbf{Input}:}{\textbf{Output}:}} 
	 \renewcommand{\algorithmicensure}{\wordbox[l]{\textbf{Output}:}{\textbf{Output}:}}
	\begin{algorithmic}
		
		\REQUIRE  Cake-cutting problem with piecewise constant valuations.
		\ENSURE  A robust proportional allocation.
	\end{algorithmic}
	\algsetup{linenodelimiter=\,}
	  \begin{algorithmic}[1] 
		
	\FOR{each $\pi\in \Pi^N$}
	\STATE $C \longleftarrow[0,1]$ (intervals left)
	\FOR{$i=1$ to $n$}
	\STATE $X_{\pi(i)}^{\pi}\longleftarrow$ maximum preference cake piece of size $1/n$ from $C$
	\STATE $C\longleftarrow C - X_{\pi(i)}^{\pi}$.
	\STATE $i\longleftarrow i+1$.
	\ENDFOR
	\ENDFOR
	
	\STATE Construct a disjoint and exhaustive interval set $\mathcal{J'}$ induced by the discontinuities in agent valuations and the cake cuts in the $n!$ cake allocations. 
	\STATE $Y_i \longleftarrow$ empty allocation for each $i\in N$.
	\FOR{each $J_j = [a_j, \ b_j] \in \mathcal{J'}$}
	\FOR{each $i\in N$}
\STATE Let $p_{ij} = \frac{count(i,J_j)}{n!}$ where $count(i,J_j)$ is the number of permutations in which $i$ gets $J_j$.
\STATE Generate $A_{ij} \subseteq J_j$, which is of length $p_{ij}|J_j|$ according to some subroutine.
\STATE $Y_i\longleftarrow Y_i \cup A_{ij}$ 
	\ENDFOR
\ENDFOR

	  \RETURN $Y=(Y_1,\ldots, Y_n)$
	 \end{algorithmic}
	\label{algo:CSD}
	\end{algorithm}

	Although CSD does not necessarily require $n!$ cuts of the cake, it may take exponential time if the number of agents is not a constant. In Example~\ref{example:CSD}, we illustrate how CSD works.

	\begin{example}[Illustration of CSD]\label{example:CSD}
		We implement CSD on the cake cutting problem in Figure~\ref{figure:piecewise}.
		%
		%
For permutation $12$, agent $1$ first chooses the cake piece $\{[0,0.1],[0.6,1]\}$ and agent $2$ then takes the remaining piece $\{[0.1,0.6]\}$. For permutation $21$, agent $2$ first chooses the cake piece $\{[0.3,0.8]\}$ and agent $1$ then takes the remaining piece $\{[0,0.3],[0.8,1]\}$.
		
		The set of all relevant subintervals induced by the two permutations are $\{[0,0.1], [0.1,0.3],[0.3,0.6],[0.6,0.8],[0.8,1]\}$. When we we additionally consider the discontinuities in the players' valuations, the set of relevant subintervals is 
		$\mathcal{J'}=\{[0,0.1], [0.1,0.3],[0.3,0.5],[0.5,0.6], [0.6,0.8],[0.8,1].$
	
		\noindent Then \[Y_1=\{[0,0.1],\frac{1}{2}[0.1,0.3],\frac{1}{2}[0.6,0.8], [0.8,1]\}\] 
	and
	\[Y_2=\{\frac{1}{2}[0.1,0.3],[0.3,0.5],[0.5,0.6], \frac{1}{2}[0.6,0.8]\}. \]
			\end{example}

\begin{proposition}\label{prop:CSD-well-defined}
	For piecewise constant valuations, CSD is well-defined and returns a feasible cake allocation in which each agents gets a piece of size $1/n$. 
	\end{proposition}
	
	\begin{proposition}\label{prop: CSD is robust prop}
		For piecewise constant valuations, CSD  satisfies robust proportionality and symmetry.
		\end{proposition}

%
%


 Notice that unlike CRSD, CSD interprets the probability of allocating each interval to an agent as allocation a fractional portion of the interval to that agent.  Unless the actual way of allocating the fractions is specified, one cannot discuss the notion of strategyproofness for CSD because a deviating agent is unable to properly evaluate his allocation against his true valuation function in the reported profile. Contrary to intuition, CSD may or may not be strategproof depending on how the fractional parts of each interval are allocated. In fact, the following remark illustrates this issue.

\begin{remark}\label{remark:csd-not-sp}
CSD is not strategyproof if the fraction of each interval of $\mathcal{J'}$ is allocated deterministically (please see the appendix for the proof). 
\end{remark}

In light of difficulty, we will implement a method (see Algorithm~\ref{CSD: subroutine}) that randomly allocating the fractions of intervals to agents. With this implementation, CSD is strategyproof in expectation.

\begin{algorithm}[t]
	  \caption{A subroutine that converts fractional allocation into subintervals via randomization}
	  \begin{algorithmic}[1] 
\STATE Generate $U_j \sim unif[a_j, \ b_j]$.
\STATE For $a_j \leq x \leq 2b_j - a_j$, let $\mod(x) = x$ if $a_j \leq x \leq b_j$ and $x - (b_j-a_j)$ if $x > b_j$.  Let  
\[A_{ij} = [\mod(U_j + \sum_{k = 1}^{i-1}p_{kj}(b_j-a_j)), \mod(U_j + \sum_{k = 1}^{i}p_{kj}(b_j-a_j))]\]
if $\mod(U_j + \sum_{k = 1}^{i-1}p_{kj}(b_j-a_j)) \leq \mod(U + \sum_{k = 1}^{i}p_{kj}(b_j-a_j))$ and
\[A_{ij} = [a_j, \mod(U_j + \sum_{k = 1}^{i}p_{kj}(b_j-a_j))] \cup [\mod(U_j + \sum_{k = 1}^{i-1}p_{nj}(b_j-a_j)), b_j]\]
otherwise.
\end{algorithmic}
\label{CSD: subroutine}
\end{algorithm}
We will refer this randomized implemention of CSD as \emph{Constrained Mixed Serial Dictatorship}, or \emph{CMSD} for short.


\begin{proposition}\label{CSD: SP}
CSD implemented with the aforementioned random allocation rule is strategyproof in expectation.
\end{proposition}

Although CSD is strategyproof in expectation, it fails to satisfy truthfulness based on group-based deviations no matter how the fractional parts of each interval are allocated.



\begin{proposition} \label{CSD: not GSP for PWC}
	For cake cutting with piecewise constant valuations, CSD is not weakly group-strategyproof even for two agents.
	\end{proposition}
%
%
%
%
%

Moreover, for cake cutting with piecewise uniform valuations, CSD is not weakly group-strategyproof for at least seven agents. The statement above follows from the fact that RSD is not weakly group-strategyproof for dichotomous preferences when there are at least seven agents~\citep{BoMo01b,BoMo04a}.

	Even though CSD satisfies both proportionality and symmetry, it does not satisfy the stronger notion of envy-freeness.

	\begin{proposition}\label{prop: CSD is not envy-free}
		CSD is not necessarily envy-free for three agents even for piecewise uniform valuations.
	\end{proposition}



Another drawback of CSD is that it is not Pareto optimal for piecewise constant valuations. The statement follows from the fact that RSD is not SD-efficient~\citep{BoMo01a}. However for the case of two agents, it is robust envy-free and polynomial-time.

\begin{proposition}
		For two agents and piecewise constant valuations, CSD is robust envy-free, and polynomial-time but not Pareto optimal. 
	\end{proposition}

\begin{remark}
	For piecewise uniform valuations, CSD can be modified to be made Pareto optimal. The main change is that for each permutation CSD, the resultant outcome needs to be made Pareto optimal. This can be done by using the idea in \citep{ABH11b}.
\end{remark}

	\section{Extensions}\label{sec:extensions}

	In this section, we show how some of our positive results concerning CCEA extend to more general settings where agents may have variable claims or they may have initial endowments (please see Algorithm~\ref{algo:ccea-extensions}).
	
	\begin{algorithm}[h]
	  \caption{ Controlled Cake Eating Algorithm (CCEA) to compute a robust envy-free allocation for piecewise constant value functions. CCEA works for either private endowments or variable claims but is not defined for both.}
	  \label{RAR-algo}
	\renewcommand{\algorithmicrequire}{\wordbox[l]{\textbf{Input}:}{\textbf{Output}:}} 
	 \renewcommand{\algorithmicensure}{\wordbox[l]{\textbf{Output}:}{\textbf{Output}:}}
	\begin{algorithmic}
		
		\REQUIRE  Piecewise constant value functions with priority claims $(c_1,\ldots, c_n)$ or private endowments $(\omega_1,\ldots, \omega_n)$ 
		\ENSURE  A robust envy-free and individually rational allocation.
	\end{algorithmic}
	\algsetup{linenodelimiter=\,}
	  \begin{algorithmic}[1] 
			
	\STATE  Let $N'=N\cup\{n+1\}$ where $n+1$ is the agent owning all public cake but with no interest in any of the cake.  Join the segments $\omega(1), \ldots, \omega(n), \omega(n+1)$ to assemble a cake. 
	\STATE Divide the regions according to agent value functions. Let $\mathcal{J}=\{J_1, \ldots, J_m\}$ be the set of subintervals of $[0,1]$ formed by consecutive marks.
	\STATE Consider $(N',H,\succsim ,e, rate(\cdot), size(\cdot))$ where

	\begin{itemize}	
	\item $H=\{h_1,\ldots, h_m\}$ where $h_i=J_i$ for all $i\in \{1,\ldots, m\}$ 
	\item  $\pref$ is defined as follows: $h\succsim_i h'$ if and only if $v_i(x)\geq v_i(y)$ for $x\in h$ and $y\in h'$; \\
	$\pref_{n+1}$ is defined in the way that each house in the market is unacceptable to agent $n+1$.
	\item $size(h_j)=1$ for $h_j\in \arg \max_{j\in \{1,\ldots, m\}}(len(J_j))$;\\

	 $size(h_j)=\frac{len(J_j)}{(\max_{j\in \{1,\ldots, m\}}(len(J_j)))}$ \\for all $h_j\notin \arg \max_{j\in \{1,\ldots, m\}}(len(J_j))$;
	\item $rate(i)=claim(i)$; 
	\item $e_{ih_j}=size(h_j)$ if $J_j\in \omega(i)$ and zero otherwise.
	\end{itemize}

	\STATE $p\longleftarrow CC(N',H,\succsim ,e, rate(\cdot), size(\cdot))$ 
	\STATE For interval $J_j$, agent $i$ is allocated a subinterval of $J_j$, denoted by $J^{i}_j$, which is of length $p_{ih_j}/size(h_j)\times len(J_j)$. \\
	$X_i\longleftarrow \bigcup_{j=1}^m p_{ih_j}/size(h_j)J_j$ for all $i\in N$\\
	
	  \RETURN $X=(X_1,\ldots, X_n)$
	 \end{algorithmic}
	\label{algo:ccea-extensions}
	\end{algorithm}

	\subsection{Variable claims}
	
We consider the concept of \emph{variable claims} in cake cutting.
	It is generally assumed in the literature that each agent has equal claim to the cake. In this case we modify the definition of proportionality and envy-freeness accordingly. Reasoning about claims has a long tradition in the fair division literature. However this strand of research assumes that agents have uniform and identical preferences over the whole of the divisible object and their only concern is the proportion of the object they get~\citep[see \eg][]{Thom03a,Youn94a}. 
\citet{BrTa96a} refer to variable claims as entitlements and touch upon entitlements in cake cutting at a few places in their book. They present a general idea of handling entitlements by cloning agents~\citep[Page 152, ][]{BrTa96a} but doing so can lead to an exponential blowup of time and space.   \citet{BJK08a} also considered entitlements in pie-cutting but presented an impossibility result and also a positive result for two agents. On the other hand, one of our algorithms handles variable claims for arbitrary number of agents and does not require cloning.
	
          An allocation satisfies  \emph{proportionality for variable claims} if $V_i(X_i) \geq c_i/\sum_{j \in N} c_j V_i([0,1])$ for all $i \in N$.
	An allocation satisfies \emph{robust proportionality for variable claims} if for all $i \in N$ and for all $v_i'\in \hat{V_i}$, $\int_{X_i}v_i'(x)dx\geq (c_i/ \sum_{j \in N} c_j)\int_{[0,1]}v_i'(x)dx$.

	An allocation satisfies \emph{envy-freeness for variable claims}, if  $V_i(X_i)\geq (c_i/c_j)V_i(X_j)$ for all $i,j\in N$. 
	
	An allocation satisfies \emph{robust envy-freeness for variable claims} if for all $i, j \in N$ and for all $v_i'\in \hat{V_i}$,
	$\int_{X_i}v_i'(x)dx\geq (c_i/c_j)\int_{X_j}v_i'(x)dx$.

		\begin{proposition}\label{prop: CCEA envy for claims}
	For cake cutting with variable claims and piecewise constant valuations, CCEA satisfies robust envy-freeness for variable claims.
	\end{proposition}

For uniform valuations, group-strategyproofness of CCEA even holds when agents have variable claims.

	\begin{proposition}\label{prop:sp-with-speed}
	For cake cutting with piecewise uniform value functions and variable claims, CCEA is group-strategyproof. 
	\end{proposition}

	On the other hand, it can be shown that the following natural modification of CSD does not satisfy proportionality in the presence of variable claims: each agent $i$ gets $c_i/(\sum_{j\in n}c_j)$ length of the most preferred interval for each permutation.

	We note that MEA can be easily adapted to cater for variable claims. Each agent's budget is proportional to his claim.
	
				\begin{proposition}\label{th:maintheorem3b}
			For cake cutting with with piecewise constant valuations and variable claims, there exists an algorithm (MEA)  that is polynomial-time, Pareto optimal, envy-free, and proportional.
	\end{proposition}
	
	%
	%
	%

	\subsection{Private endowments}
	
In the cake cutting literature, the cake is generally considered to be a public good which is divided among the agents. It could also be that each agent brings part of the cake and then the assembled cake needs to be reassigned among the agents fairly. We call the setting cake cutting with \emph{private endowments}. In this context, individual rationality is a minimal requirement and envy-freeness need to be redefined. Classic cake cutting can be modelled by cake cutting with private endowments in the following manner: none of the agents have any endowments and one additional agents owns the whole cake but has no interest in the cake.
	
	An allocation is \emph{individually rational} if $V_i(X_i)\geq V_i(\omega(i))$. An allocation satisfies \emph{justified envy-freeness for private endowments} if $V_i(X_i)\geq V_i(X_j)$ or $V_j(X_i)< V_j(\omega(j))$ for each pair of agent $i,j\in N$. Either $i$ should not be envious of $j$ or if he is envious, then giving $i$'s allocation to $j$ violates individual rationality.\footnote{The definition of justified envy-freeness for private endowments is based on a similar concept due to \citet{Yilm09a} for the domain of housing assignment problems.} It is not entirely clear how to obtain the canonical concept of justified envy-freeness which takes into account both private endowments and variable claims.
	
		An allocation satisfies \emph{robust justified envy-freeness for private endowments} if for all $i,j\in N$ either
	$\forall v_i'\in \hat{V_i} \midd \int_{X_i}v_i'(x)dx\geq \int_{X_j}v_i'(x)dx$ or $ \exists v_j'\in \hat{v_j}\midd \int_{X_i}v_j'(x)dx<\int_{\omega(j)}v_j'(x)dx$. If $i$ does not get an unambiguously as preferred a piece as $j$'s, $j$ has a counter claim that $j$'s endowment is better than $i$'s piece for some ordinally equivalent valuation function of $j$. 
	
		\begin{proposition}\label{prop:CCEA endowments}
	For cake cutting with private endowments and piecewise constant valuations, CCEA satisfies robust justified envy-freeness for private endowments.
	\end{proposition}

	Generalizing CSD to handle endowments without losing its properties seems to be a challenging task. We have not addressed it in this paper.
		\section{Related work and discussion}

	A mathematical analysis of cake cutting started with the work of Polish mathematicians Steinhaus, Knaster, and Banach~\citep{Stei48a}. As applications of fair division have been identified in various multiagent settings, a topic which was once considered a mathematical curiosity has developed into a full-fledged sub-field of mathematical social sciences~\citep[see \eg][]{Moul03a}. 
	In particular, in the last few decades, the literature of cake cutting has grown considerably~\citep[see \eg][]{Moul03a,BrTa96a,RoWe98a,Proc12a}. 

%
%
	%

	The relation between the random assignment problem and cake cutting has been noticed before~\citep{CLPP12a}. However, in their discussion of related work, \citeauthor{CLPP12a} argue that techniques from the random assignment literature cannot be directly applied even to piecewise uniform functions---a subclasses of piecewise constant functions. 
	Stumbling blocks identified  include the fact that in the random assignment problem, each agents gets one object. 
	We observed that PS satisfies envy-freeness even when agents get multiple objects.\footnote{Even in the first paper on the PS algorithm for strict preferences, it was observed that PS can be extended to the case where there are many more objects than agents~\citep[Page 310, ][]{BoMo01a}.} 
	Another issue raised is that `if two agents desire two subintervals, both agents would value the longer subinterval more than the shorter.' 
This problem can be circumvented by using the algorithm to compute the greatest common denominator of the lengths of the intervals and then dividing the intervals in equally sizes subintervals. However the number of subintervals can be exponential in the input size. 
We use the idea of \citet{AtSe11a} that some houses may only be partially available. 

	\citet{CLPP12a} stated that generalizing their strategyproof algorithm for piecewise uniform valuations to the case for piecewise constant valuations as an open problem.
We presented two algorithms --- CCEA and MEA --- that generalize Mechanism 1 of \citet{CLPP10a, CLPP12a}. Although they both satisfy certain desirable properties, both natural generalizations are not strategyproof.



	%



         A number of works in the cake-cutting literature reason about strategyproofness. However they refer to a very weak notion of strategyproofness which is equivalent to a maximin strategy. Apart from the papers of \citet{CLPP10a,CLPP12a}, we are aware of no positive results regarding discrete, strategyproof, and fair algorithms even for the restricted domain of piecewise constant valuations. In this paper we present a proportional algorithm (CSD) for piecewise constant valuations that although not formally strategyproof seems less like likely to manipulate in comparison to CCEA. If we are allowed to use randomization, then we show that CSD can be adapted to be strategyproof in expectation. Notice that if we instead require our algorithm to be strategyproof always, and proportional in expectation, then CRSD would satify these properties. Moreover, we note that \citet{CLPP10a} also presented an algorithm that is randomized, envy-free and proportional, and strategyproof in expectation. However, it does not satisfy unanimity and requires much more randomization in contrast to CSD.  It remains an open question whether there exists an algorithm that always returns a proportional allocation and is always immune to agent manipulation. 
	
	A difficulty that arises in coming up with strategyproof and proportional algorithm lies in that there is no restriction on the distribution of the discontinuity points of the agents' valuation functions. To illustrate this point, suppose the algorithm designer knows that the discontinuity points of the agents' valuation functions come from a set $S = \{x_1, \ldots, x_k\}$, where $0 \leq x_1 \leq \ldots \leq x_k \leq 1$. Consequently, a mechanism that partitions $[0,1]$ into intervals of the form $[x_i, x_{i+1}]$ and allocates $1/n$ of each interval to each agent would be proportional, envy-free and strategyproof. Even if the designer does not know such a $S$, but instead we require the minimum distance between any two consecutive discontinuity points of the agent's valuation function to be at least some $\epsilon > 0$, then there exists a strategyproof and $\delta$-proportional algorithm for this setting.

		As for private endowments, we are only aware of \citep{Thom07a} where private endowment are considered in pie-cutting. However, in comparison to our positive results concerning envy-freeness, \citet{Thom07a} presents a negative result concerning the core.
	
	%
	%
	%
	
	\section{Conclusion}

In this paper, we made a number of conceptual and technical contributions. It will be interesting to consider new results with respect to robust versions of fairness, endowments, and variable claims. Cake cutting is an exciting sub-field of microeconomics with numerous applications to computer science. In order for theory to be more relevant to practice, we envision exciting work in much richer models of cake cutting.

\section*{Acknowledgements}
The authors acknowledge the helpful comments  of Jay Sethuraman and Simina Br{\^a}nzei.

This material is based upon work supported by
	the Australian Government's
	Department of Broadband, Communications and the Digital
	Economy, the Australian Research Council, the Asian
	Office of Aerospace Research and Development through
	grant AOARD-12405.


 
\newpage

\normalsize
\appendix

\section*{Proof of Theorem~\ref{PE fair impossibility}}
\begin{proof}
Consider the following two-agent profile.

Agent 1: 
\[v_1(x) = v^{1}_a \ \text{for} \ x \in [0,0.5], \ v_1(x) = v^{1}_b \ \text{for} \ x \in (0.5,1].\]

Agent 2: 
\[v_2(x) = v^{2}_a \ \text{for} \ x \in [0,0.5], \ v_2(x) = v^{2}_b \ \text{for} \ x \in (0.5,1].\]

Choose $v^{1}_a, v^{1}_b,  v^{2}_a, v^{2}_b > 0$ in such a way that $v^{1}_a > v^{1}_b$ and $v^{2}_a > v^{2}_b$ and $\frac{v^{1}_a}{v^{1}_b} > \frac{v^{2}_a}{v^{2}_b}$. By either robust envy-freeness or robust proportionality, the mechanism must make the following allocation $x$ where $x^{1}_a = x^{1}_b = x^{2}_a = x^{2}_b = 0.25$. On the other hand, in order for the mechanism to be Pareto efficient,  $x$ must be an element of $P^{1} = \{x | x^{1}_{b} = 0 \ \text{or} \ x^{2}_a = 0\}.$ Hence, we have reached an impossibility.
\end{proof}

\section*{Proof of Theorem~\ref{PE SP Prop impossibility}}
\begin{proof}
For cake cutting with piecewise constant valuations and $n \geq 2$, it follows from \citep[Theorem 3, ]{Schu97a} that the only type of strategyproof and Pareto optimal mechanisms are dictatorships. Consequently, there exists no strategyproof and Pareto optimal mechanism that is also proportional or symmetric.
\end{proof}

\section*{Proof of Theorem~\ref{SP Rob Prop Non W impossibility}}
\begin{proof}
The following example shows that there exists no cake-cutting algorithm that is strategyproof, robust proportional, and non-wasteful. \\

Profile 1:

\[1: v_1(x) = a  \ \text{if} \ x \in [0,0.5], \ v_1(x) = b \ \text{if} \ x \in [0.5,1]\]
\[2: v_2(x) = a \ \text{if} \ x \in [0,0.5] , \ v_2(x) = b \ \text{if} \ x \in [0.5,1]\]

for some $a > b > 0$. \\

Since the algorithm is robust proportional, it must be the case that each agent 
receive $1/2$ of $[0,0.5]$ and $1/2$ of $[0.5, 1]$. This is because the fractional parts of $[0,0.5]$ and $[0.5,1]$ that each agent receives must stochastically dominate the uniform allocation, otherwise the allocation would not satisfy proportionality for some valuation function in the ordinal equivalence class.

Without lost of generality (up to reordering and regrouping of the cake), we may assume that agent $1$ receives $[0, 0.25] \cup [0.5, 0.75]$ and agent $2$ receives $[0.25, 0.5] \cup [0.75, 1]$. \\

Now consider profile 2:

\[1: v_1(x) =  a \ \text{if} \ x \in [0,0.25] , \ v_1(x) = b \ \text{if} \ x \in [0.5,0.75], \ v_1(x) = 0  \ \text{otherwise}\]
\[2: v_2(x) =  a \ \text{if} \ x \in [0,0.5] , \ v_2(x) = b \ \text{if} \ x \in [0.5,1]\]

By SP, agent $1$ must again receive $[0, 0.25] \cup [0.5, 0.75]$ and agent 2 
receives $[0.25, 0.5] \cup [0.75, 1]$. If agent $1$ receives anything less in 
profile $2$, then he would deviate from profile $2$ to profile $1$. If agent $1$ 
receives anything more in profile $2$, then he would deviate from profile $1$ to 
profile $2$. \\

Now consider profile 3:

\[1: v_1(x) = a \ \text{if} \ x \in [0,0.25], v_1(x) = b \ \text{if} \ x \in [0.5,0.75], \ v_1(x) = 0  \ \text{otherwise} \]
\[2: v_2(x) = a + \epsilon \ \text{if} \ x \in [0,0.25], v_2(x) = a \ \text{if} \ x \in [0.25,0.5], v_2
(x) = b \ \text{if} \ x \in [0.5,1]\]

By robust proportionality, both agent $1$ and $2$ must receive $1/2$ of $[0,0.25]$. By non-wastefulness, agent $2$ must receive $[0.25, 0.5]$ and $[0.75,1]$ since agent $1$ has a utility of $0$ on these intervals. Hence, agent $2$ in profile 2 would misreport so that he receives the allocation in profile $3$.

\end{proof}

\section*{Proof of Proposition~\ref{prop:CCEA is robust envy-free}}

\begin{proof}[Sketch]
First of all, CCEA is non-wasteful because an agent is never allowed to eat a piece of the cake that he has no desire for. On the other hand, the algorithm terminates only when every portion of the cake that is desired by at least one agent is completely consumed by some agent who desires it.

Next, we show that the algorithm is robust envy-free. Consider a fractional assignment $p$ returned by the CC algorithm. Without private endowments CC is equivalent to the EPS algorithm of \citet{KaSe06a}. Assignment $p$ satisfies justified envy-freeness in presence of variable eating rates: $u_i(p_i) \geq {u_i(p_{j})}$ for all utilities $u$ consistent with preferences of $i$ over the houses. The intuition is that at any point during the running of CC, an agent $i$  will be `eating' his most favoured object(s) at the same rate as any other agent $j$ even if $j$ is also the eating the same object(s). Hence, for all $v_i'\in \hat{V_i}$, it is the case that for $j\neq i$, $\int_{X_i}v_i'(x)dx\geq \int_{X_j}v_i'(x)dx$. \\
\end{proof}

\section*{Proof of Proposition~\ref{prop:CCEA time}}

\begin{proof}[Sketch]For a cake cutting instance $I$, $|I|=nm$  is the input size where $n$ is the number of agents and $m$ is the number of relevant subintervals. Once the lengths of the subintervals in $\{J_1, \ldots, J_m\}$ are computed, the size of each house can be computed in linear time. The number of houses in CCEA is $m$. We now analyse the running time of CC on $n$ agents and $m$ houses~\citep[Section 3.5, ][]{AtSe11a}.
In the CC algorithm, the flow network consists of $|V|$ vertices and $|E|$ arcs where $|V|=O(n^2)$ and $|E|=O(n^2m)$. The number of parametric flow problems needed to be solved is $O(nm)$. A parametric flow network problem can be solved in time $O(|V||E|\log({|V|}^2/|E|))$ due to \citet{GGT89a}. Hence, the running time of CCEA is $O(nm(n^2)(n^2m)\log(n^4/n^2m))=O(n^5m^2\log(n^2/m))$.
%
\end{proof}

\section*{Proof of Proposition~\ref{prop:CCEA is not SP}}

	\begin{proof}[Sketch]
		CCEA is not strategyproof even if all the piecewise intervals are of equal length, and there are no private endowments or variable claims, and agents have strict preferences over the intervals. In this case CCEA is equivalent to the classic PS algorithm. It is known that PS is not strategyproof even for strict preferences when there are more objects than agents~\citep{Koji09a}. 
		%
		%
		%
		%
	\end{proof}

\section*{Proof of Lemma~\ref{prop: equivalence}}


 \begin{proof}[Sketch]In the absence of private endowments and variable claims, CCEA can be solved by invoking EPS instead of CC but with the slight modification that in the corresponding flow network of EPS, the capacity of each arc $(h,t)$ is set to $size(h)$ in step 2 of EPS \citep[Algorithm 1, ][]{KaSe06a}. Let us refer to this simplified CCEA as SimpleCCEA. When SimpleCCEA is run, it invokes EPS and solves repeated 
 parametric network flow problems \citep[Step 3, Algorithm 1, ][]{KaSe06a}. 
 In the step, EPS computes a bottleneck set of agents and houses at each break-point. SimpleCCEA computes bottleneck sets in the same way as Mechanism 1 of \citet{CLPP10a} and then allocates the resources in the bottleneck sets to the agents in the bottleneck set. The flow networks of the slightly modified EPS \citep[Figure~2, ][]{KaSe06a} and that of Mechanism 1~\citep[Figure~2, ][]{CLPP10a} are identical with only two insignificant differences namely that in the flow network of Mechanism 1 of \citet{CLPP10a} i) the source and target are swapped and all the arcs are inverted and ii) the size of the houses/intervals is not normalized. However, the eventual allocations are same.
	 \end{proof}

\section*{Proof of Proposition~\ref{prop:gsp1}}

\begin{proof}
In light of lemma \ref{prop: equivalence}, it suffices to show that SimpleCCEA is GSP. We begin with some notations. \\ Let $len(X)$ denote the length of $X \subseteq [0,1]$. Since the utility function is piecewise uniform, it suffices to consider the length of pieces of the cake that are desired by each agent. \\
Let $S \subseteq N$ be a coalition of agents who misreport their true preference. \\
Let $I$ denote the instance where every agent reports truthfully and $I'$ denote the instance where agents in $S$ misreport. \\
Let $D_1,\ldots,D_n \subseteq [0,1]$ denote the pieces of cake that are truly desired by each agent. \\
Let $D'_1,\ldots,D'_n \subseteq [0,1]$ denote the desired pieces of cake that are reported by each agent. \\
Let $A_1,\ldots,A_n \subseteq [0,1]$ denote the allocation received by each agent under truthful reports. \\
Let $A'_1,\ldots,A'_n \subseteq [0,1]$ denote the allocation received by each agent when the agents in $S$ misreport. \\
Let $X_1, \ldots, X_k$ be the bottleneck sets of agents with respect to the true preferences of agents arranged in the order that they are being allocated by the algorithm in instance $I$. \\
Let $X'_1, \ldots, X'_p$ be the bottleneck sets of agents with respect to the true preferences of agents arranged in the order that they are being allocated by the algorithm in instance $I'$. \\
Let $len(X_i)$ denote the length of the allocation each agent receives in the bottleneck set $X_i$.
Let 
\[\overset{\sim}{X_l} = \{i \in X_l \mid len(A'_i \cap D_i) \geq len(A_i \cap D_i) = len(A_i)\} \footnote{The fact that $ len(A_i \cap D_i) = len(A_i)$ makes use of the free disposal property, i.e. the allocation that the mechanism gives agent $i$ is a subset of the pieces desired by agent $i$  under truthful reports.}\]
\[\overset{\wedge}{X_l} = \{i \in X_l \mid len(A'_i \cap D_i) \leq len(A_i)\}\]
In other words $\overset{\sim}{X_l}$ is the subset of agents of $X_l$ who weakly gain in utility when the agents in $S$ misreport, and  $\overset{\wedge}{X_l}$ is the subset of agents of $X_l$ who weakly lose in utility when the agents in $S$ misreport. We will show that for all $l = 1,\ldots,k$,  $X_l = \overset{\wedge}{X_l}$. This would then directly imply that $S$ must be empty since no coalition $S$ can exist such that  by misreporting, everyone in the coalition is weakly better off and at least one agent in the coalition is strictly better off. 

We will prove this result via induction on $k$. In order to carry on with the induction, we will show that no agent in $X_1$ appears in the coalition $S$. We begin with a lemma.

\begin{lemma}
It is the case that $X_1 = \overset{\sim}{X_1}$. In other words, no agent in $X_1$ is strictly worse off when some subset of agents misreport their preference.
\end{lemma}

It is clear that $\overset{\sim}{X_1} \subseteq X_1$. Suppose the claim does not hold, then there exists some agent $i \in X_1 \backslash \overset{\sim}{X_1}$. In other words, agent $i$ is strictly worse off due to the misreports. Since an agent would only misreport his preference if misreporting weakly improves his utility, we deduce that $i \notin S$, which implies that $D_i = D'_i$. Consequently,  the following set of inequalities hold for agent $i$:
\[ len(A'_i) = len(A'_i \cap D'_i) = len(A'_i \cap D_i)  < len(A_i), \]
where the first equality follows from the fact that $A'_i \subseteq D'_i$ by the way the algorithm allocates the pieces of cake to agents. The second equality follows from $D_i = D_i'$. And the last inequality follows from $i \in X_1 \backslash \overset{\sim}{X_1}$. \\
We claim that since $X_1$ is the first bottleneck set with respect to the true preference, it cannot be the case that $ len(X'_l) = len(A'_i) < len(A_i) = len (X_1)$, where $X'_l$ is the bottleneck set that $i$ belongs to in the instance $I'$. Suppose this is the case, then we have that $len(X'_1) \leq len(X'_l) < len(X_1)$.\footnote{It is shown in \citep{CLPP12a} that the length of allocation of agents weakly increases with respect to the index of bottleneck sets.} It is clear that $X'_1$ cannot contain an agent who misreports in $I'$, since a misreporting agent in $X'_1$ only receives a piece of cake with length $len(X'_1) < len(X_1)$, which is strictly less than what he would've gotten had he reported truthfully. Hence, every agent in $X'_1$ must report his true preference in $I'$. Since $X'_1$ is the first bottleneck set in $I'$, we have that 
\[
\begin{split}
len(X'_1) = \frac{len((\cup_{j \in X'_1} D'_j) \cap [0,1])}{|X'_1|} &=  \frac{len((\cup_{j \in X'_1} D_j)\cap [0,1])}{|X'_1|} \\ 
\geq \frac{len((\cup_{j \in X_1} D_j)\cap [0,1])}{|X_1|} &= len(X_1),
\end{split}
\]

where the first and third equalities follows from the way the algorithm makes an allocation, the second equality follows from the fact that the agents in $X'_1$ have the same reports in $I'$ as in $I$, and the inequality follows from the fact that $X_1 \in \arg\min_{S \subseteq N} \frac{len((\cup_{j \in S} D_j \cap [0,1]))}{|S|}$. This contradicts the fact that $len(X'_1) < len(X_1)$.

\begin{lemma}
It is the case that $X_1 = \overset{\wedge}{X_1} =  \overset{\sim}{X_1}$. In other words, no agent in $X_1$ is strictly better off when some subset of agents misreport their preference.
\end{lemma}

Since we have established that $X_1 = \overset{\sim}{X_1}$, and $\overset{\wedge}{X_1} \subseteq X_1$, it suffices to show that $ \overset{\sim}{X_1} \subseteq \overset{\wedge}{X_1}$. Suppose not, then for all $i \in X_1$, we have that $ len(A'_i \cap D_i) \geq len(A_i)$ and there exists some $j \in X_1$ such that $ len(A'_i \cap D_j) > len(A_j)$. Summing over $i \in X_1$, we get that 
\[len(\cup_{i \in X_1} (A'_i \cap D_i))  = \sum_{i \in X_1} len(A'_i \cap D_i) > \sum_{i \in X_1} len(A_i) = len(\cup_{i \in X_i} A_i) = len(\cup_{i \in X_i} D_i),\]
where the first two equalities follow from the fact that the $A_i$'s and $A'_i \cap D_i$'s are disjoint subsets and the third equality follows from the way the algorithm allocates to the agents in the smallest bottleneck set. But this set of inequalities contradict the fact that $\cup_{i \in X_1} (A'_i \cap D_i) \subseteq \cup_{i \in X_i} D_i $, which implies that $len(\cup_{i \in X_1} (A'_i \cap D_i)) \leq len(\cup_{i \in X_i} D_i)$.   Hence, it must be the case that for every $i \in X_1$, we have that $ len(A'_i \cap D_i) = len(A_i)$, which implies that $i \in \overset{\wedge}{X_1}$.

\begin{lemma}
No agent in $X_1$ appears in the coalition $S$ and $X_1$ is also the first bottleneck set for $I'$.
\end{lemma}
By the previous lemma, no agent in $ \overset{\wedge}{X_1}$ is strictly better off by misreporting his preference. Thus, any agent in $X_1$ would potentially be in $S$ if by misreporting, he makes himself no worse off and simultaneously make some other agent in $S$ strictly better off. The only way that this can happen is by misreporting, the agents in $X_1$ make their collective claim over their desired pieces smaller, so that agents in later bottleneck set can claim some of their desired pieces. On the other hand, the following inequality
\[len(\cup_{i \in X_1} D'_i ) \geq len(\cup_{i \in X_1} A'_i ) \geq len(\cup_{i \in X_1} (A'_i \cap D_i)) = len(\cup_{i \in X_i} D_i)\]
implies that if any subset of agents of $X_1$ wants to misreport so they are not worse off, then collectively, they must over-report their preference to obtain allocations that together is weakly larger in total than the allocations they would get had they reported truthfully. Thus, having a subset of agents in $X_1$ misreport will not benefit the other agents in the coalition $S$. Hence, we may conclude that no agent in $X_1$ appears in the coalition $S$.  Provided that every agent in $X_1$ also reports truthfully in $I'$, there is no incentive for an agent that belongs to a subsequent bottleneck set in $I$ to misreport and prevent $X_1$ from being the first bottleneck set in $I'$ since that would make the misreporting agent strictly worse off, as in doing so, he needs to create a bottleneck set $X'$ such that $len(X') < len(X)$ and he would consequently receive an allocation of $len(X')$.\\

Since no agent in $X_1$ appears in the coalition $S$ and $X_1$ is also the first bottleneck set for $I'$, we can remove $X_1$ from $N$ and $\cup_{i \in X_1} A_i$ from $[0,1]$ and do induction on the set of remaining agents $N \backslash X_1$ and the remaining piece of cake $[0,1] \backslash \cup_{i \in X_1} A_i$ to be allocated and the proof is complete by invoking the inductive hypothesis with $X_2$ being the first bottleneck set in the new instance.
\end{proof}

\section*{Proof of Proposition~\ref{PE, EF, Prop}} 
\begin{proof}
Consider the following math program

\begin{align}
\text{max} \quad
\prod_{j =1}^{n} u_j \notag \\
\text{s.t.} \quad
u_j = \sum_{i=1}^{k} v_{ij}x_{ij} &\quad \forall j = 1, \ldots, n \notag \\
\sum_{j=1}^{n} x_{ij} \leq l_i &\quad \forall i = 1, \ldots, k \notag \\
x_{ij} \geq 0 &\quad \forall i, j. \notag \\  \notag
\end{align}
where $l_i = len(J_i)$, \\

Notice that the feasible region of the math program contains all feasible allocations. An optimal solution given by the LP is not Pareto dominated by any other feasible allocation because that would contradict the optimality of the solution. Hence, it is Pareto efficient. \\

To see that the optimal solution of the math program is also an envy free allocation, if we instead view $x_{ij}$ as the fractional amount of $J_i$ that is allocated to agent $j$, then scaling the $v_{ij}$'s appropriately (i.e. setting $v'_{ij} = v_{ij}l_i$), then solving the math program stated in the proposition is equivalent to solving the following math program.
\begin{align}
\text{max} \quad
\prod_{j =1}^{n} u_j \notag \\
\text{s.t.} \quad
u_j = \sum_{i=1}^{k} v'_{ij}x_{ij} &\quad \forall j = 1, \ldots, n \notag \\
\sum_{j=1}^{n} x_{ij} \leq 1 &\quad \forall i = 1, \ldots, k \notag \\
x_{ij} \geq 0 &\quad \forall i, j. \notag \\  \notag
\end{align}

which in turn equivalent to solving

\begin{align}
\text{max} \quad
\sum_{j =1}^{n} \log {u_j} \notag \\
\text{s.t.} \quad
u_j = \sum_{i=1}^{k} v'_{ij}x_{ij} &\quad \forall j = 1, \ldots, n \notag \\
\sum_{j=1}^{n} x_{ij} \leq 1 &\quad \forall i = 1, \ldots, k \notag \\
x_{ij} \geq 0 &\quad \forall i, j. \notag \\  \notag
\end{align}

Notice that the above math program is a convex program since we are maximizing a concave function (or equivalently minimizing a convex function) subject to linear constraints. \\

In \citep[pp 105-107, ][]{Vazi07a}, Vazirani invites us to consider a market setting of buyers (agents) and divisible goods (intervals). Each good is assumed to be desired by at least one buyer (i.e. for every good $i$, $v_{ij} > 0$ for some buyer $j$). There is a unit of each good and each buyer is given the same amount of money say $1$ dollar, for which he uses to purchases the good(s) that maximizes his utility subject to a set of given prices. The task is to find a set of equilibrium prices such that the market clears (meaning all the demands are met and no part of any item is leftover) when the buyers seek purchase good(s) to maximize their utility given the equilibrium prices. \\

Using duality theory, one can interpret the dual variable $p_i$ associated with the constraints $\sum_{j=1}^{n} x_{ij} \leq 1$ as the price of consuming a unit of good $i$. By invoking the KKT conditions, \citet{Vazi07a} shows the prices given by the optimal dual solution is a unique set of equilibrium prices. Moreover, the primal optimal solution for each buyer $j$ is precisely the quantity of good(s) that the buyer ends up purchasing that maximizes his utility given the equilibrium prices. \\

Now we can argue as to why the optimal primal solution is an envy free allocation. Because given the equilibrium prices, if a buyer desires another buyer's allocation, since he has the same purchasing power as any other buyer, he would instead use his money to obtain the allocation of the buyer that he envies. This would result in some surplus and deficit of goods, contradicting the fact that the given prices are equilibrium prices.   

\end{proof}

\section*{Proof of Proposition~\ref{CCEA equivalence}}
\begin{proof}
For piecewise uniform valuations, it is known that CCEA is equivalent to mechanism $1$ in \citep{CLPP12a}, which we will refer to as \emph{SimpleCCEA}. The remainder of the proof will focus on showing the equivalence between SimpleCCEA and the convex program for piecewise uniform valuations. To do so, given an allocation of SimpleCCEA, which is a feasible solution of the convex program, we will find a set of prices corresponding to the allocation and show that the prices are in fact the equilibrium prices defined by Vazirani on pages 105-107 of \cite{Vazi07a}. Moreover, this allocation would be an allocation that maximizes the agents' utility given the equilibrium prices. \\

Using the same notations as those in \citep{CLPP12a}, given a valuation profile, let $B$ be the set of buyers/agents and $G$ be the set of goods or intervals.  Let $S_i$ be the $i$-th bottleneck set computed by SimpleCCEA, i.e. 
$S_i = S_{min}  \in \arg \min_{S' \subseteq S} avg(S', X_i)$ in the $i$-th iteration of the subroutine of SimpleCCEA. 

Let $G_i$ be the set of goods that are distributed amongst the buyers in $S_i$. In the convex program, since each buyer is endowed with $1$ dollar and every buyer in $S_i$ receives $avg(S_i, X_i)$ units of good(s), it is natural to define the price of a unit of each good $k \in G_i$ to be 
\[p_k = \frac{1}{avg(S_i, X_i)} = \frac{|S_i|}{len(D(S_i, X_i))}.\] 
Notice that the prices for each good is well defined (i.e. each good has exactly one nonnegative price). This follows from the following observations:
\begin{enumerate}
\item $\cup G_i = G$ (or every good has at least one price). This follows from the assumption that every good is desired by at least one agent, which means that SimpleCCEA will allocate all of the goods.
\item $G_i \cap G_j = \emptyset$ for all $i \neq j$ (or every good has at most one price). This follows from the fact that no fractional parts of any good is allocated to agents from two or more bottleneck sets, which another algorithmic property of SimpleCCEA.
\end{enumerate}

To show that the $p_k$'s form a set of equilibrium prices, we will show that given the $p_k$'s, the buyers in every $S_i$ will choose to purchase \emph{only} goods from $G_i$ to maximize their utility function. We will do induction on the number of bottleneck sets. Consider the first bottleneck set $S_1$, we will show that 
\begin{lemma}$G_1$ are the only items that are desirable by buyers in $S_1$. 
\end{lemma}
\begin{proof}
SimpleCCEA finds $S_1$ by solving a parametrized max flow problem on a bipartite network. The network has a node for every buyer and a node for every good. There is a directed edge from a buyer $i$ to a good $j$ with infinite capacity if good $j$ is desired by buyer $i$. In addition, there is a source $s$ and a sink $t$. There is a directed edge from $s$ to each buyer $i$ with capacity $\lambda$ and a directed edge from edge good $j$ to $t$ with capacity equaling the quantity of good $j$. $\lambda$ is set to $0$ initially so that the unique min cut in the network for $\lambda = 0$ is $\{s, B \cup G \cup t\}$. $\lambda$ is gradually raised until $\{s, B \cup G \cup t\}$ is no longer the unique min cut, at which point $S_1$ and $G_1$ are found by looking for another min cut of the form $\{s \cup S_1 \cup G_1, B \backslash S_1 \cup G \backslash G_1 \cup t\}$. Since $\{s \cup S_1 \cup G_1, B \backslash S_1 \cup G \backslash G_1 \cup t\}$ is a min cut, it must be the case that $G_1 = nbr(S_1)$ in the network, which proves the lemma. (If $nbr(S_1) \backslash G_1 \neq \emptyset$, then there will be an infinite capacity edge crossing the cut. On the other hand, if $ G_1 \backslash nbr(S_1) \neq \emptyset$, then replacing $G_1$ with $nbr(S_1)$ on the $s$ side of the cut would give a cut with a smaller capacity.)
\end{proof}
Since the agents have piecewise uniform valuation, in the setting of the convex program, each buyer in $S_1$ has the same utility for the items that he desire in $G_1$. Moreover, given that the prices of goods are identical in $G_1$, each buyer is indifferent between choosing among his desirable items for the best item in terms of bang per buck. Hence, for all buyers in $S_1$, the allocation given by SimpleCCEA maximizes buyer's utility given the prices $p_k$'s. Moreover, notice that the money of buyers in $S_1$ and goods in $G_1$ are exhausted by the allocation given by SimpleCCEA. \\

After the goods in $G_1$ are allocated to the buyers in $S_1$ we repeat the same argument for the remaining buyers and goods and the inductive hypothesis allows to conclude that for every $i$ and for all buyers in $S_i$, the allocation given by SimpleCCEA maximizes buyer's utility given the prices $p_k$'s. Moreover, the money of buyers in $S_i$ and goods in $G_i$ are exhausted by the allocation given by SimpleCCEA.  There is a slight difference between the inductive step and the base case, as it is possible that some buyer $b$ in $B \backslash S_i$ also desires certain goods in $G_j$ for some $j < i$. However, \citet{CLPP12a}  state that $avg(S_i, X_i)$ is a weakly increasing function of $i$, which means that $p_i = 1/avg(S_i, X_i)$ is a weakly decreasing function of $i$. Since we are dealing with piecewise uniform valuations, the utility of a buyer over his desirable goods are identical, this means that for any buyer in $S_i$, the goods that maximize his bang for buck are in $G_i$. \\

Putting everything together and we have shown that $p_k$'s constitute the set of equilibrium prices and SimpleCCEA gives an equilibrium allocation, which is an optimal solution to the convex program.
\end{proof}

\section*{Proof of Proposition~\ref{prop: uniform allocation is not SP}}
Consider the following profile of two agents.

Profile 1:
\begin{itemize}
\item[]$v_1(x) = 1 \ \text{if} \ x \in [0, 0.2], \ v_1(x) = 0 \ \text{if} \ x \in (0.2, 1].$
\item[]$v_2(x) = 0 \ \text{if} \ x \in [0,0.6],  \ v_2(x) = 1 \ \text{if} \ x \in (0.6, 1].$
\end{itemize}

\noindent The uniform allocation rule gives us the allocation:  
\begin{itemize}
\item[]$Y_1: \{\frac{1}{2}[0,0.2],  \frac{1}{2}(0.6,1]\}$. \\
\item[] $Y_2:  \{\frac{1}{2}[0,0.2],  \frac{1}{2}(0.6,1]\}$.
\end{itemize}
Let $A \subset (0.6,1]$ be the allocation that agent $2$ receives in this case. \\

Now consider profile 2:

\begin{itemize}
\item[]$v_1(x) = 1 \ \text{if} \ x \in [0, 0.2], \ v_1(x) = 0 \ \text{if} \ x \in (0.2, 1].$
\item[]$v_2(x) = 0 \ \text{if} \ x \in [0,0.1] \backslash A,  \ v_2(x) = 1 \ \text{if} \ x \in A.$
\end{itemize}

\noindent The uniform allocation rule gives us the allocation:  
\begin{itemize}
\item[]$Y_1: \{\frac{1}{2}[0,0.2],  \frac{1}{2}A\}$. \\
\item[] $Y_2:  \{\frac{1}{2}[0,0.2],  \frac{1}{2}A\}$.
\end{itemize}

\noindent Hence, agent $2$ in profile $2$ would misreport so that the reported profile is profile $1$.

\section*{Proof of Proposition~\ref{prop:CSD-well-defined}}
%
\begin{proof}

			We first prove that CSD is well-defined and results in a feasible allocation in which each agent gets $1/n$ size of the cake. Let $\mathcal{J'}=\{J_1,\ldots, J_{\ell}\}$ be a partitioning of the interval $[0,1]$ induced by the discontinuities in agent valuations and the cake cuts in the $n!$ cake allocations. We make a couple of claims about $\mathcal{J'}$ that following from the way 
	$\mathcal{J'}$ is constructed. 
	
	\begin{claim}
An agent is completely indifferent over each subinterval in $\mathcal{J'}$.
	\end{claim}
	
	
		\begin{claim}
Let $X_i^{\pi}$ denote a maximum preference cake piece of size $1/n$ chosen by agent $i$ in the serial order $\pi$. For each $J\in \mathcal{J'}$ either $X_i^{\pi}$ contains $J$ completely or it does not contain any part of $J$.
		\end{claim}

		Now consider a  matrix of dimension $n! \times \ell$:  $B=(b_{ij})$ such that $b_{ij}=1$ if $J_j\subset X_i^{\pi}$ and $b_{ij}=0$ if $J_j\nsubset X_i^{\pi}$. Since for each $\pi\in \Pi^N$, each agent $i\in N$ gets $1/n$ of the cake in $X_i^\pi$, then it follows that $\sum_{i=1}^{n!}\sum_{j=1}^{\ell} b_{ij}{len(J_j)}=n!/n$. Hence, 
		\begin{align*}
		\sum_{j=1}^{\ell} \sum_{i=1}^{n!}b_{ij}{len(J_j)}/n!=1/n. 
		\end{align*}

Also consider a  matrix of dimension $n \times \ell$:  $M=(m_{ij})$ such that $m_{ij}$ denotes the fraction of $J_j$ that agent $i$ gets in $Y_i$. From the algorithm CSD, we know that $m_{ij}=\frac{count(i,J_j)}{n!}J$ where $count(i,J_j)$ is the number of permutations in which $i$ gets $J_j$. It is immediately seen that each column sums up to $1$. Hence each $J_j$ is complete allocated to the agents. We now prove that each agent gets a total cake piece of size $1/n$. We do so by showing that $\sum_{j=1}^{\ell} m_{ij} len(J_j)=1/n$. 
		\begin{align*}
			1/n &= \sum_{j=1}^{\ell} \sum_{i=1}^{n!}b_{ij}{len(J_j)}/n!
			= \sum_{j=1}^{\ell}(\sum_{i=1}^{n!}b_{ij})len(J_j)/n!
			= \sum_{j=1}^{\ell}(count(i,J_j))len(J_j)/n!\\
			&= \sum_{j=1}^{\ell}(\frac{count(i,J_j)}{n!})len(J_j)
			=\sum_{j=1}^{\ell} m_{ij} len(J_j).
\end{align*}

Hence $Y=(Y_1,\ldots, Y_n)$ the allocation returned by CSD is a proper allocation of the cake in which each agent gets a total cake piece of size $1/n$.
\end{proof}

\section*{Proof of Proposition~\ref{prop: CSD is robust prop}}
	\begin{proof}

		We first argue for proportionality of CSD. In the case where all agents have the same valuations as the valuation of $i$, $i$ is guaranteed $1/n$ of the value of the whole cake because of anonymity of CSD. First note that for each $\pi\in \Pi^N$ and preferences $V_{-i}$ of all agents other than $i$,
	$
		V_i(\text{CSD}^{\pi}(V_i,V_{-i}))\geq V_i(\text{CSD}^{\pi}(V_i,(V_i,\ldots, V_i))).
	$

		The reason is that when valuations are not identical, predecessors of $i$ in $\pi$ leave weakly better cake for $i$ as when their valuations are same as agent $i$. Hence, 
	$
		V_i(\text{CSD}(V_i,V_{-i}))\geq V_i(\text{CSD}(V_i,(V_i,\ldots, V_i)))=V_i([0,1])/n.
	$


			Finally, note that when an agent selects his best possible cake piece in each permutation, the exact height of the valuation function is not relevant and only the relative height matters. Hence, CSD in fact satisfies robust proportionality. Symmetry for CSD follows directly from symmetry for RSD.
	\end{proof}



\section*{Proof of Remark~\ref{remark:csd-not-sp}}

CSD would not be strategyproof if the fraction of each subinterval of $\mathcal{J'}$ is allocated deterministically and the allocation is made public information before the agents submit their valuation function. To see this, consider the following example of two agents.

Profile 1:
\begin{itemize}
\item[]$v_1(x) = 1 \ \text{if} \ x \in [0, 0.5], \ v_1(x) = 0 \ \text{if} \ x \in (0.5, 1].$
\item[]$v_2(x) = 1 \ \text{if} \ x \in [0, 0.5],  \ v_2(x) = 0 \ \text{if} \ x \in (0.5, 1].$
\end{itemize}

Running CSD gives us the allocation:  
\begin{itemize}
\item[]$Y_1: \{\frac{1}{2}[0,0.5],  \frac{1}{2}(0.5,1]\}$. \\
\item[] $Y_2:  \{\frac{1}{2}[0,0.5],  \frac{1}{2}(0.5,1]\}$.
\end{itemize}

Profile 2:
\begin{itemize}
\item[]$v_1(x) = 1 \ \text{if} \ x \in [0, 0.5], \ v_1(x) = 0 \ \text{if} \ x \in (0.5, 1].$
\item[]$v_2(x) = 0  \ \text{if} \ x \in [0,0.25], \ v_2(x) = 1 \ \text{if} \ x \in (0.25, 0.75], v_2(x) = 0  \ \text{if} \ x \in (0.75, 1].$
\end{itemize}

Running CSD gives us the allocation:  
\begin{itemize}
\item[]$Y_1 = \{[0,0.25], \frac{1}{2}(0.25,0.5], \frac{1}{2}(0.75,1]\}$. \\
\item[] $Y_2 = \{\frac{1}{2}[0.25,0.5], (0.5,0.75],  \frac{1}{2}(0.75,1]\}$.
\end{itemize}

Now it is possible that the CSD mechanism decides to gives $[0,0.25] \cup (0.75,1]$ to agent 1 and $[0.25,0.75]$ to agent 2 in profile 1. Consequently, knowing this, agent 2  in profile 2 would misreport his valuation function to be $v_2(x) = 1 \ \text{if} \ x \in [0, 0.5], v_2(x) = 0 \ \text{if} \ x \in (0.5, 1]$ in order to receive the allocation given in profile 1 and gain utility in doing so.\\


\section*{Proof of Proposition~\ref{CSD: SP}}

Consider the profiles $P$ and $P^i$, where $P$ is a profile where every agent reports truthfully and $P^i$ is a profile where agent $i$ misreports while fixing every other agent's report to be the same as that in $P$. Let $\sigma$ denote a permutation of $[n] = \{1,\ldots,n\}$ and let $S$ denote the set of all permutations of $[n]$.  Let $J_1,\ldots, J_k$ denote the intervals whose fractional allocations are specified to each agent by CSD in profile $P$ and $J'_1,\ldots, J'_{k'}$ denote the intervals  whose fractional allocations are specified to each agent by CSD in profile $P^i$. Let $V_i(J)$ denote agent $i$'s total utility derived from receiving the interval $J$ and $V_i(A(\sigma))$ and $V_i(A'(\sigma))$ denote the agent $i$'s total utility derived from his allocated pieces when the serial ordering of the agents is $\sigma$ in profile $P$ and $P^i$ respectively. Let $p_{ij}$ denote the probability that interval $J_j$ is assigned to agent $i$. Since random serial dictatorship is strategyproof in expectation, we have that
\[\sum_{j=1}^{k} p_{ij} V_i(J_j) = \sum_{\sigma \in S} \frac{1}{n!} V_i(A(\sigma)) \geq \sum_{\sigma \in S} \frac{1}{n!} V_i(A'(\sigma)) = \sum_{j=1}^{k'} p_{ij} V_i(J'_j) \]
Now CSD views $p_{ij}$ as allocating a $p_{ij}$ fraction of interval $J_j$ to agent $i$. In order for a deviating agent to properly evaluate the utility derived from his allocation in the deviating profile $P^i$, we have to come up with an allocation rule that actually \emph{attains} the utility $\sum_{j=1}^{k'} p_{ij} V_i(J'_j)$ for agent $i$ (either deterministically or in expectation) when the profile of reports is $P^i$. In particular, say if we want to allocate a subinterval of $J_i$ with length $p_{ij}$ times that of $J_i$ to agent $i$ at random, then this random allocation rule must satisfy the property that $\sum_{j=1}^{k'}E[V_i(A_{ij})] = p_{ij}V_i(J'_j)$, where $A_{ij} \subset J'_j$ such that $|A_{ij}| = p_{ij}|A_{ij}|$.  In order to do so, we will show that the randomized allocation rule for CSD satisfy the property that every agent $i$ and interval $J_j$, we have that $E[V_i(A_{ij})] = p_{ij}V_i(J'_j)$ for all valuation functions $V_i$, where  $A_{ij} \subset J'_j$ such that $|A_{ij}| = p_{ij}|A_{ij}|$. \\

Notice that $mod(U_j + \sum_{k = 1}^{i-1}p_{nj}(b_j-a_j))$ is uniformly distributed on $[a_j,b_j]$ and $A_{ij}$ has length $p_{ij}(b_j-a_j)$. The fact that $E[V_i(A_{ij})] = p_{ij}V_i(J'_j)$ follows from the following lemma.

\begin{lemma}\label{lemma: CSD-random}
Let $U$ be uniformly distributed on the interval $[a,b]$ and let $0 \leq \alpha \leq 1$. Let $A = [U, U + \alpha (b-a)]$ if $U + \alpha (b-a) \leq b$ and $A = [a, U - (1 - \alpha)(b-a)] \cup [U, b]$ if $U + \alpha (b-a) > b$, then we have that 
$E_U[V_i(A)] = \alpha V_i([a,b])$, where $V_i(X) = \int_{X} v_i(x) dx$ for any integrable function $v_i$.
\end{lemma}
\begin{proof}
Define $\bar{v_i} = v_i$ for $a \leq x \leq b$ and $\bar{v_i}(x) = \bar{v_i}(x+(b-a))$ for $x$ outside $[a,b]$. Since $\bar{v_i}$ is periodic, then it suffices to show that 
\[ E_U[V_i(A)] = \int_{a}^{b} \int_{x}^{x+\alpha (b-a)} \bar{v_i}(y) dy \frac{1}{b-a} dx = \alpha \int_{a}^{b} \bar{v_i}(x) dx = \alpha V_i([a,b]).\]
By drawing a picture of the region that we evaluate the integral over and due to the fact that $\bar{v_i}$ is periodic, we have that 
\[ \frac{1}{b-a} \int_{a}^{b} \int_{x}^{x+\alpha (b-a)} \bar{v_i}(y) dy dx =   \frac{1}{b-a} \int_{a}^{b}\bar{v_i}(y) \int_{y-\alpha (b-a)}^{y}  dx dy =\alpha \int_{a}^{b} \bar{v_i}(x) dx,\]
which proves the lemma.
\end{proof}

\section*{Proof of Proposition~\ref{CSD: not GSP for PWC}}

\begin{proof}
Let $a = [0,0.25], b = (0.25, 0.5], c = (0.5, 0.75], d=(0.75,1]$. Consider the following two profiles of valuations. \\
Profile 1:
\begin{itemize}
\item[]$v_1(x) = 4 \ \text{if} \ x \in a, \ v_1(x) = 3 \ \text{if} \ x \in b, \ v_1(x) = 2 \ \text{if} \ x \in c, \ v_1(x) = 1 \ \text{if} \ x \in d.$
\item[]$v_2(x) = 3  \ \text{if} \ x \in a, \ v_2(x) = 4  \ \text{if} \ x \in b, \ v_2(x) = 1 \ \text{if} \ x \in c, \ v_2(x) = 2 \ \text{if} \ x \in d.$
\end{itemize}

Running CSD gives us the allocation:  
\begin{itemize}
\item[] $Y_1 = \{\frac{1}{2}a , \frac{1}{2}b , \frac{1}{2}c, \frac{1}{2}d\}$. \\
\item[] $Y_2 = \{\frac{1}{2}a, \frac{1}{2}b, \frac{1}{2}c, \frac{1}{2}d\}$.
\end{itemize}

Profile 2:
\begin{itemize}
\item[]$v_1(x) = 4 \ \text{if} \ x \in a, \ v_1(x) = 2 \ \text{if} \ x \in b, \ v_1(x) = 3 \ \text{if} \ x \in c, \ v_1(x) = 1 \ \text{if} \ x \in d.$
\item[]$v_2(x) = 2  \ \text{if} \ x \in a, \ v_2(x) = 4  \ \text{if} \ x \in b, \ v_2(x) = 1 \ \text{if} \ x \in c, \ v_2(x) = 3 \ \text{if} \ x \in d.$
\end{itemize}

Running CSD gives us the allocation:  
\begin{itemize}
\item[]$Y_1 = \{a,  c\}$. \\
\item[] $Y_2 = \{b, d\}$.
\end{itemize}
Hence, agents with true valuation in profile 1 would misreport together to profile 2, which means that CSD is not group strategyproof for 2 agents.
\end{proof}

\section*{Proof of Proposition~\ref{prop: CSD is not envy-free}}

\begin{proof}
	There are three agents, each with piecewise uniform valuation function.
	$v_1(x) = 1.5$ for $x \in [0, 2/3]$ and $0$ otherwise. \\
	$v_2(x) = 1.5$ for $x \in [0, 1/3] \cup (2/3,1]$ and $0$ otherwise. \\
	$v_3(x) = 1.5$ for $x \in (1/3, 1]$ and $0$ otherwise. \\
 
	Let $a = [0,1/3]$, $b = (1/3, 2/3]$, $c = (2/3, 1]$. 

	We adopt the following implementation of CSD: when it is agent i's turn to pick, out of the pieces of the remaining cake that he likes, he takes the \emph{left-most} such piece with length 1/n, where n is the number of agents.

	If the priority ordering were $1,2,3$, then a feasible assignment that respects the preferences is 1 $\leftarrow$ a, 2 $\leftarrow$ c, 3 $\leftarrow$ b.

	If the priority ordering were 1,3,2, then a feasible assignment that respects the preferences is $1 \leftarrow a,$ $3 \leftarrow b$, $2 \leftarrow c$.

	If the priority ordering were $2,1,3$, then a feasible assignment that respects the preferences is $2 \leftarrow a$, $1 \leftarrow b$, $3 \leftarrow c$.

	If the priority ordering were $2,3,1$, then a feasible assignment that respects the preferences is $2 \leftarrow a$, $3 \leftarrow b$, $1 \leftarrow c$.

	If the priority ordering were $3,1,2$, then a feasible assignment that respects the preferences is $3\leftarrow b$, $1 \leftarrow a$, $2 \leftarrow c$.

	If the priority ordering were $3,2,1$, then a feasible assignment that respects the preferences is $3 \leftarrow b$, $2 \leftarrow a$, $1 \leftarrow c$.
	\noindent
	Then, the CSD allocation is as follows.
	\begin{align*}
		Y_1 = &\quad \{\frac{1}{2}[0,1/3], \frac{1}{6}(1/3, 2/3], \frac{1}{3}(2/3, 1]\}\\
		Y_2 = &\quad \{\frac{1}{2}[0,1/3], \frac{1}{2}(2/3,1]\}\\
		Y_3 = &\quad \{\frac{5}{6}(1/3, 2/3], \frac{1}{6}(2/3,1]\}
	\end{align*}
	
	%

	Clearly, agent 1 envies agent 3 in this case.
\end{proof}

\section*{Proof of Proposition~\ref{prop: CCEA envy for claims}}

\begin{proof}[Sketch]Consider a fractional assignment $p$ returned by the CC algorithm. Without private endowments CC is equivalent to the EPS algorithm of \citet{KaSe06a}. Assignment $p$ satisfies justified envy-freeness in presence of variable eating rates: $u_i(p_i) \geq (c_i/c_j){u_i(p_{j})}$ for all utilities $u$ consistent with preferences of $i$ over the houses. The informal intuition is that at any point during the running of CC, an agent $i$ with a higher eating rate than $j$ will be `eating' his most favoured object(s) faster than $j$ even if $j$ is also has the eating the same object(s). Hence, for all $v_i'\in \hat{V_i}$, it is the case that for $j\neq i$, $\int_{X_i}v_i'(x)dx\geq (c_i/c_j)\int_{X_j}v_i'(x)dx$. 
\end{proof}

\section*{Proof of Proposition~\ref{prop:sp-with-speed}}

\begin{proof}[Sketch]
Given a variable claim instance $I$ of $n$ agents where agent $i$ has claim rate $c_i$ for $i = 1, \ldots, n$. We may assume without lost of generality that the claim rates are integral. If they are not, then we can simply multiple each claim rate $c_i$ by a common denominator to make each $c_i$ integral. Doing so will not change the allocation given by the algorithm since only relative claim rates matter to the algorithm. \\
       Now consider a cake cutting instance $I'$ of $\sum_{i=1}^{n} c_i$ agents, where the agents have piecewise uniform utility function and there are no private endowments or variable claims. Moreover, for every $i = 1,\ldots,n$, there are $c_i$ agents in $I'$ each of whom has the same utility function as that of agent $i$ in $I$. It is not difficult to see that if one aggregates the allocation that the $c_i$ agents in $I'$ who share agent $i$'s valuation in $I$, then one would get an equivalent allocation (in terms of utility) to agent $i$'s allocation.\\ 
Suppose for the sake of contradiction that CCEA is not GSP for the case of variable claims, then in some instance $I$, there exists some coalition $S$ of the agents that weakly gains in utility by misreporting their preference. Now consider the equivalent instance $I'$ with no variable claims under the aforementioned transformation, then there exists some coalition $S'$ of the agents in $I'$ that weakly gains in utility by misreporting their preference, which implies that CCEA is not group-strategyproof for the no variable claims case, contradicting the result of Proposition~\ref{prop:gsp1}.
\end{proof}

\section*{Proof of Proposition~\ref{th:maintheorem3b}}

\begin{proof}
The allocation can be obtained by solving the following convex program.
\begin{align}
	\text{min} \quad
	-\sum_{j =1}^{n} c_jlog(u_j) \notag \\
	\text{s.t.} \quad
	u_j = \sum_{i=1}^{k} v_{ij}x_{ij} &\quad \forall j = 1, \ldots, n \notag \\
	\sum_{j=1}^{n} x_{ij} \leq l_i &\quad \forall i = 1, \ldots, k \notag \\
	x_{ij} \geq 0 &\quad \forall i, j. \notag \notag
	\end{align}

The proof of the desired properties is similar to the case where $c_j = 1$.
\end{proof}

\section*{Proof of Proposition~\ref{prop:CCEA endowments}}

\begin{proof}[Sketch]Consider a fractional assignment $p$ returned by the CC algorithm.
We know that $p$ satisfies justified envy-freeness for the random/fractional assignment problem~\citep[Prop. 4, ][]{AtSe11a}. If $p_i \mathrel{\succsim_i^{SD}} p_{j}$, then $\forall v_i'\in \hat{V_i} \midd \int_{X_i}v_i'(x)dx\geq \int_{X_j}v_i'(x)dx$. If $\neg(p_i \mathrel{\succsim_{j}^{SD}} e_{j})$, then $ \exists v_j'\in \hat{v_j}\midd \int_{X_i}v_j'(x)dx<\int_{\omega(j)}v_j'(x)dx$.  Hence $X$ satisfies justified envy-freeness for private endowments. 
\end{proof}

\end{document}